\documentclass[11pt,envcountsame]{article}

\usepackage{fullpage}
\usepackage{ucs}
\usepackage[utf8x]{inputenc}
\usepackage[hypertexnames=false]{hyperref}
\usepackage[T1]{fontenc}
\usepackage{listings}
\usepackage{microtype}
\usepackage{lmodern}
\usepackage{graphicx}
\usepackage{makeidx}
\usepackage{amssymb}
\usepackage{appendix}
\usepackage{pdfpages}

\usepackage{amsmath}
\usepackage{amsthm}
\usepackage{calc}
\usepackage{tabularx}
\usepackage{algpseudocode}
\usepackage{algorithm}
\usepackage{cleveref}
\usepackage{thmtools}
\usepackage{thm-restate}

\newtheorem{definition}{Definition}

\newtheorem{lemma}[definition]{Lemma}
\newtheorem{theorem}[definition]{Theorem}

\newtheorem{corollary}[definition]{Corollary}
\newtheorem{example}[definition]{Example}
\newtheorem{conjecture}[definition]{Conjecture}

\usepackage{newtxtext}
\usepackage[hmargin=1in,vmargin=1in]{geometry}
\usepackage{authblk}
\usepackage{mathtools}

\usepackage{tikz}
\usepackage{tensor}
\usetikzlibrary{arrows,positioning}

\newcommand{\sat}[1]{\ensuremath{\Gamma^{\scriptscriptstyle #1}_{\mathit{\scriptscriptstyle \mathrm{SAT}}}}}

\newcommand{\overbar}[1]{\mkern 0.8mu\overline{\mkern-0.8mu#1\mkern-0.8mu}\mkern 0.8mu}

\newcommand{\cclone}[1]{\ensuremath{\langle #1 \rangle}}
\newcommand{\strongof}[1]{\ensuremath{[#1]_s}}

\newcommand{\str}[1]{\ensuremath{\mathrm{Str}}}
\newcommand{\pcclone}[1]{\ensuremath{\langle #1 \rangle_{\not \exists}}}
\newcommand{\ar}{\ensuremath{\mathrm{ar}}}

\newcommand{\eq}{\ensuremath{\rm{Eq}}}
\newcommand{\ppol}{\ensuremath{\rm{pPol}}}
\newcommand{\pol}{\ensuremath{\rm{Pol}}}
\newcommand{\inv}{\ensuremath{\rm{Inv}}}

\newcommand{\N}{\ensuremath{\mathbb{N}}}

\newcommand{\domain}{\ensuremath{\mathrm{domain}}}

\newcommand{\CSP}{\protect\ensuremath\problemFont{CSP}}
\newcommand{\SAT}{\protect\ensuremath\problemFont{SAT}}

\newcommand{\problemFont}[1]{\textsc{#1}}

\newcommand{\padd}[1]{\ensuremath{{\cal P}_R}}
\newcommand{\upadd}[1]{\ensuremath{{\cal UP}_P}}

\newcommand{\pro}{\ensuremath{\mathrm{Proj}}}
\newcommand{\malt}{\ensuremath{\phi}}

\newcommand{\near}{\ensuremath{\mathrm{nu}}}
\newcommand{\universal}{\ensuremath{\mathrm{u}}}

\newcommand{\edge}{\ensuremath{e}}

\newcommand{\concat}[4]{\ensuremath{#1 \tensor[^{}_{#3}^{}_{}]{\times}{^{}_{#4}^{}_{}}#2}}

\begin{document}
\title{Which NP-Hard SAT and CSP Problems Admit 
  Exponentially Improved Algorithms?}
\author[1]{Victor Lagerkvist\thanks{victor.lagerqvist@tu-dresden.de}}
\author[2]{Magnus Wahlstr\"om\thanks{magnus.wahlstrom@rhul.ac.uk}}
\affil[1]{\small Institut f\"ur Algebra, TU Dresden, Dresden, Germany}
\affil[2]{\small Department of Computer Science, Royal Holloway,
  University of London, Great Britain}
\date{}
\pagenumbering{gobble}

\maketitle
\begin{abstract}
We study the complexity of $\SAT(\Gamma)$ problems for potentially infinite languages $\Gamma$
closed under variable negation, which we refer to as \emph{sign-symmetric} languages $\Gamma$.
Via an algebraic connection, this reduces to the study of restricted partial polymorphisms
we refer to as \emph{pSDI-operations} (for partial, self-dual and idempotent),
under which the language $\Gamma$ is invariant.  First, we focus on the language classes
themselves.  We classify the structure of the least restrictive pSDI-operations,
corresponding to the most powerful languages $\Gamma$, and find that these operations
can be divided into \emph{levels}, corresponding to a rough notion of difficulty,
where every level $k$ has an easiest language class, containing the language for $(k-1)$-SAT,
and a hardest language class, containing (among other things)
constraints encoded as roots of multivariate polynomials of degree $(k
- 1)$. 
Particular classes in each level correspond to the natural partially defined versions
of previously studied total algebraic invariants. In particular, the easiest class
on level $k\geq 3$ corresponds to the partial $k$-ary \emph{near-unanimity} ($k$-NU) operation,
and a larger class corresponds to the partial \emph{$k$-edge} operation.
The largest class at each level corresponds to a partial operation $\universal_{k}$ we call
\emph{$k$-universal}. Furthermore, every sign-symmetric language $\Gamma$ not
preserved by $\universal_{k}$ implements all $k$-clauses,
hence $\SAT(\Gamma)$ is at least as hard as $k$-$\SAT$; and if
$\Gamma$ is not preserved by $\universal_k$ for any $k$, 
then $\SAT(\Gamma)$ is trivially SETH-hard (i.e., takes time $O^*(2^n)$ under SETH). 

Second, we consider implications of this for the complexity of $\SAT(\Gamma)$.
We find that particular classes in the hierarchy correspond to previously
known algorithmic strategies. In particular, languages preseved by the
partial 2-edge operation can be solved via
\textsc{Subset Sum}-style meet in the middle,
and languages preserved by the partial 3-NU operation can be solved 
via fast matrix multiplication.
These results also hold for the correspondning non-Boolean CSP problems.
We also find that \emph{symmetric} 3-edge languages reduce to finding a monochromatic 
triangle in an edge-coloured graph, which can be done 
using algorithms for sparse matrix multiplication;
and if the \emph{sunflower conjecture} holds for sunflowers with $k$ petals,
then the partial $k$-NU language has an improved algorithm via
Sch\"oning-style local search.

Complementing this, we show a lower bound, showing that for every level $k$
there is a constant $c_k$ such that for every partial operation $p$ on level $k$, 
the problem $\SAT(\Gamma)$ with $\Gamma=\inv(p)$
cannot be solved faster than $O^*(c_k^n)$ unless SETH fails.
In particular, when $\Gamma=\inv(\text{2-edge})$,
this gives us the first NP-hard SAT problem which simultaneously
has non-trivial upper and lower bounds on the running time, assuming SETH.
Finally, we note a possible conjecture: It is consistent with our
present knowledge that $\SAT(\Gamma)$ admits an improved algorithm
if and only if $\Gamma$ is preserved by $\universal_k$ for some
constant $k$. However, to show this in the positive poses
some significant difficulty.
\end{abstract}

\newpage
\pagenumbering{arabic}

\section{Introduction}
\label{sec:intro}
Significant attention has been paid to the exact time complexity of
SAT and its various restrictions; in particular CNF-SAT and $k$-SAT,
but also other restrictions such as \textsc{Not-All-Equal SAT},
\textsc{1-in-$k$ SAT}, and several more
cases~\cite{CyganDLMNOPSW16TALG,hertli2014,impagliazzo2001,LokshtanovPTWY17SODA,DBLP:conf/coco/SchederS17}.
The usual focus is on an improved algorithm for some particular
variant, i.e., showing that the problem can be solved in time
$O^*(c^n)$ for some $c<2$, or, in some cases, that such an improvement
is not feasible, up to our current knowledge (i.e., it would require
disproving the strong exponential-time hypothesis, SETH; see
below). Here, and in the sequel, the parameter $n$ will in this
context always denote the number of variables in a given instance.
But what is the general rule for when a SAT problem admits such an
improved algorithm?  And can we say anything at all about lower bounds
on such improvements?

To refine the question, let us recall some terminology. 
A \emph{constraint language} is a (possibly infinite) set $\Gamma$ of finitary relations
$R \subseteq D^{\ar(R)}$ over some domain $D$, where $\ar(R)$ denotes the arity of $R$.
We will mainly focus on the Boolean case, i.e., $D=\{0,1\}$. Then
$\SAT(\Gamma)$, occasionally called the {\em parameterized
  satisfiability problem},  is 
the SAT problem where the constraints of the instance are applications
of relations from $\Gamma$, i.e., the constraints are statements
that $R(x_1,\ldots, x_r)$ must hold, for some $R \in \Gamma$
and some variables $x_1, \ldots, x_r$ from the variable set (where we do allow 
repetitions of a variable). 
The multi-valued generalization of $\SAT$, the {\em constraint
  satisfaction problem} over $\Gamma$ ($\CSP(\Gamma)$) is defined in essentially the same way, except that
$\Gamma$ may be non-Boolean. Full definitions of the problems under
consideration follow in Section~\ref{sec:prel}.
Thus, for example, 3-SAT corresponds to $\SAT(\sat{3})$
where $\sat{3}$ for each 3-clause in $\{(x \lor y \lor z), \ldots,
(\neg x \lor \neg y \lor \neg z)\}$ contains the relation
excluding only the tuple forbidden by that particular clause.
Similarly, for $k \geq 3$ let $\sat{k}$ denote
the constraint language of all $k$-clauses, i.e.,
$\SAT(\sat{k})$ is equivalent to $k$-SAT.

Let us also tentatively define $c(\Gamma)$ as the infimum over all constants $c>1$
such that $\SAT(\Gamma)$ can be solved in $O(c^n)$ on $n$ variables.
Then the \emph{exponential time hypothesis} (ETH), due to Impagliazzo and Paturi,
states that $c(\sat{k})>1$ for every $k$, and was shown to be 
equivalent to the statement that $c(\sat{3})>1$~\cite{impagliazzo2001}. It has also been shown to be equivalent to the statement
that $c(\Gamma)>1$ for every $\Gamma$ such that $\SAT(\Gamma)$ is NP-hard~\cite{jonsson2017}.
The \emph{strong exponential time hypothesis} (SETH)
is the statement that $\lim_{k \to \infty} c(\sat{k})=2$~\cite{impagliazzo2009,impagliazzo2001}.
Then our main research question can be rephrased as, 
for which constraint languages $\Gamma$ is $c(\Gamma)<2$,
respectively, when would $c(\Gamma)<2$ contradict SETH?
We say that $\SAT(\Gamma)$ allows an \emph{improved algorithm}
in the former case, and that it is \emph{SETH-hard} in the
latter. Hence, our main interest is in exponential improvements rather
than subexponential improvements of the form $O(2^{n - o(n)})$ which
have been proven to exist for CNF-SAT~\cite{evgeny2005}. 

Before we discuss our approach for the general case, we consider a few examples. 
First of all, the algorithms for $k$-SAT imply that
$c(\Gamma)<2$ for every finite language $\Gamma$.
However, such bounds are also known for some infinite languages. One
example is
\textsc{Exact SAT}, the language  
of 1-in-$k$-clauses of all arities, which
admits an improved algorithm~\cite{exactsat}.
As has been shown more recently, so does 
the problem where constraints are encoded as
the roots of bounded-degree multivariate polynomials
over a finite field~\cite{LokshtanovPTWY17SODA}.
Thus, we need a way to discuss properties of infinite arbitrary
languages, and we need to consider the representation of 
constraints from such a language. 
We address these issues in Section~\ref{sec:results}.

Lower bounds on $c(\Gamma)$ for some $\Gamma$ have
been significantly harder to come by. Some SAT problems
have been shown to be SETH-hard, in particular 
\textsc{Not-All-Equal SAT} and problems related to SAT
such as \textsc{Hitting Set}~\cite{CyganDLMNOPSW16TALG}. It is also known
that assuming ETH, 
the value of $c(\sat{k})$ increases infinitely often~\cite{impagliazzo2001}.
However, we do not even have conjectural evidence against
any particular value of $c(\Gamma)$ for any language $\Gamma$
such that $\SAT(\Gamma)$ is not SETH-hard, other than for
trivial cases.\footnote{By trivial cases, we mean problems
where the natural search space is smaller than $2^n$ but
otherwise unrestricted. 
Consider a language where every variable
is involved in a disequality, e.g., the language of relations
$R'(x_1,\ldots, x_{2k}) \equiv (x_1 \neq x_{k+1}) \land \ldots 
 \land (x_k \neq x_{2k}) \land R(x_1,\ldots,x_k)$
for arbitrary relations $R$.
It is easy to see that under SETH, 
this problem has $c(\Gamma)=2^{1/2}$.}
We also are not aware of any previous attempts to
engage with the question of what makes a SAT
problem SETH-hard or not in general. 

In this paper, we study these questions using tools from universal algebra. 
It is known that the value of $c(\Gamma)$ is determined by
algebraic invariants of $\Gamma$ known as \emph{partial
  polymorphisms}~\cite{jonsson2017}.
It is not difficult to prove that if $\Gamma$ has no interesting partial
polymorphisms, then $\SAT(\Gamma)$ is trivially SETH-hard.
We study the converse to this question, to essentially ask,
does the existence of even a single relevant partial polymorphism~$p$
imply that $\SAT(\Gamma)$ has an improved algorithm?
In particular, is it possible to design an algorithm
with an exponentially improved running time, whose correctness
depends only on~$p$? One of the main strengths of using such an
algebraic approach is that it makes the task of identifying languages
$\Gamma$ such that $c(\Gamma) < 2$ considerably easier. In fact, as
we discuss in Section~\ref{sec:algebraic_approach}, these languages
can be succinctly classified according to the expressive power of
individual partial operations.

Our paper has two main contributions. First, we 
characterize the structure of the weakest non-trivial invariants~$p$.
In this, we restrict ourselves to \emph{sign-symmetric} languages (see
below). 
This reveals a characterization of problem complexity, with
close ties to several previously studied problems and algorithm
classes.
Second, we use the framework to provide both upper and lower bounds
on $c(\Gamma)$ for the corresponding languages $\Gamma$, under SETH. 
We show 
that algorithms from the literature
can be extended to work for every language having
a certain partial polymorphism~$p$. 
In the negative direction, we are able to prove lower bounds on
$c(\Gamma)$ for every language $\Gamma$ characterised purely by 
its invariants. 
As a result, we produce the first language $\Gamma$ such that
$c(\Gamma)$ has both non-trivial upper and lower bounds
under SETH. 
Finally, we make connections between these $\SAT(\Gamma)$ problems 
and some problems in polynomial-time fine-grained complexity.

Our approach also implies some results for CSPs on a non-Boolean
domain, but our main focus in the present paper lies in studying the
Boolean case. 


\subsection{Universal algebraic aspects of SAT problems}
\label{sec:algebraic_approach}

To make the discussion of our approach more precise, we need to review
some notions from universal algebra. This is simply intended as an
introduction and overview to make the extended abstract
self-contained; full definitions follow later in the paper in Section~\ref{sec:prel}.
The universal algebraic approach to problem complexity originates in
research into the constraint satisfaction problem
(CSP)~\cite{jeavons1997}. Recall the definitions of a constraint language $\Gamma$ and the problem
$\CSP(\Gamma)$ from the preceding section.
Clearly, the complexity of CSP$(\Gamma)$ varies as a
function of $\Gamma$: if $\Gamma$ is simple enough, then CSP$(\Gamma)$
is in P; and if $\Gamma$ is rich enough, then CSP$(\Gamma)$ is
NP-complete.  The \emph{dichotomy conjecture}, first posed by Feder
and Vardi~\cite{FV98}, states that these are
the only two cases and that no NP-intermediate CSP problems exist: for
every fixed language $\Gamma$, CSP$(\Gamma)$ is either in P or is
NP-complete. This conjecture has been the subject of intense research
and the piece remaining to complete the puzzle was recently resolved
by two independent authors~\cite{bulatov2017,zhuk2017}.

The algebraic approach turned out to be central in this research programme. In
short, this approach boils down to the realization that properties of
constraint languages can be expressed by properties of their {\em
  polymorphisms}. Informally, a polymorphism
of a constraint language $\Gamma$ is an operation which yields a
method to combine satisfying assignments of instances of CSP$(\Gamma)$. The algebraic
reformulation of the CSP dichotomy theorem then states that
CSP$(\Gamma)$ is tractable if there exists a non-trivial method to
combine solutions, and is NP-complete otherwise. 
More formally, we may define polymorphisms as follows. First, let $R \subseteq D^n$ be a relation on $D$, and let $p: D^r \to D$ be
an $r$-ary operation over $D$. We can then generalise $p$ to an operation
$(D^n)^r \to D^n$ on tuples over $D$ by
$p(x_1, \ldots, x_r)[i] = p(x_1[i], \ldots, x_r[i])$ for every position $i \in [n]$ (where $x_j[i]$ denotes the $i$th
element of the tuple $x_j$).  Then $p$ is a \emph{polymorphism
of $R$} if this generalised operation preserves $R$, i.e., if $p(x_1,
\ldots, x_r) \in R$ for any $x_1, \ldots, x_r \in R$.  Note that
if $p$ is a {\em projection}, i.e., $p(t_1,\ldots,t_r)
= t_i$ for some $i \in [r]$, then $p$ preserves every possible
relation. The notion of a polymorphism easily extends to constraint
languages, and we say that $p$ is a polymorphism of the constraint
language $\Gamma$ if $p$ is a polymorphism of $R$ for
every relation $R \in \Gamma$, and let $\pol(\Gamma)$ denote this set.
It is then known that the complexity
of CSP$(\Gamma)$, up to polynomial-time many-one reductions, is determined
entirely by $\pol(\Gamma)$~\cite{jeavons1998}. 

\begin{theorem}
  Let $\Gamma$ and $\Delta$ be finite constraint languages over a finite
  domain $D$. If $\pol(\Delta) \subseteq \pol(\Gamma)$, then $\CSP(\Gamma)$ is
  polynomial-time many-one reducible to $\CSP(\Delta)$.
\label{thm:jeavons}
\end{theorem}

At this stage this result may seem slightly puzzling since we do not
yet have a clear correspondence between polymorphisms and their
implications on constraint languages.
However, there exists a dual concept to polymorphisms on the
relational side called \emph{implementations}. 
Given a set of relations $\Gamma$ over a domain $D$, a $k$-ary relation $R$ is
definable by a {\em primitive positive implementation} over $\Gamma$ (pp-definable)
if there exists a first-order formula making use of existential
quantification and conjunctive constraints over $\Gamma$ such that the
set of models of this formula is precisely $R$. Given a constraint language
$\Gamma$ we then let $\cclone{\Gamma}$ be the smallest set of
relations containing $\Gamma$ and which is closed under taking pp-definitions.
The polymorphisms of $\Gamma$ then characterize
the power of pp-definitions over $\Gamma$ in the following sense.

\begin{theorem}[\cite{BKKR69i,BKKR69ii,Gei68}]
  Let $\Gamma$ and $\Delta$ be two constraint languages. Then 
   $\Gamma \subseteq \cclone{\Delta}$ if and only if  $\pol(\Delta) \subseteq
    \pol(\Gamma)$.
\end{theorem}

This duality has two implications. First, note that an instance 
of $\CSP(\Gamma)$ can be viewed as a special case of a pp-definition over $\Gamma$, 
hence the polymorphisms of $\Gamma$ describe closure properties 
for the whole $\CSP(\Gamma)$ problem, and can be used to design
polynomial-time algorithms. This is in line with the intuition that a
polymorphism yields a method for combining satisfying assignments. 
Second, if $R$ has a pp-definition in $\Gamma$ then there is a
polynomial-time many-one reduction from $\CSP(\Gamma \cup \{R\})$ to
$\CSP(\Gamma)$; essentially, the pp-definition describes a classical
``gadget reduction'' between the problems obtained by replacing
constraints over $R$ by the collection of constraints over $\Gamma$
prescribed by the pp-definition. Therefore, dually to the previous
point, the absence of sufficiently interesting polymorphisms for
$\Gamma$ would imply a polynomial-time reduction from an NP-hard
problem $\CSP(\Gamma')$, e.g., \textsc{3-SAT}, to $\CSP(\Gamma)$.

In practice, for CSPs beyond the Boolean domain,
the complexity landscape gets very complex and
one needs to apply a richer algebraic toolbox to make
progress. However, it was
realized early that not only does the complexity of CSP$(\Gamma)$
depend on $\pol(\Gamma)$, but in fact only the {\em identities}
satisfied by the operations in $\pol(\Gamma)$~\cite{bulatov2005}. In
technical terms this means that the complexity of $\CSP(\Gamma)$ only
depends on the {\em variety} generated by $\pol(\Gamma)$. 
 We will
 not define these concepts formally since they are not needed to
 present the main results; it is sufficient to know that the
 complexity of $\CSP(\Gamma)$ only depends on the identities satisfied
 by the operations
 in $\pol(\Gamma)$. For example, $\CSP(\Gamma)$ is
 solvable using $k$-consistency if $\pol(\Gamma)$ contains a {\em majority operation},
 i.e., a ternary operation $m$ satisfying the identities $m(x,y,y) =
 y$, $m(y,x,y) = y$, $m(y,y,x) = y$~\cite{jeavons1997}. Moreover, all operations
 resulting in tractable CSPs can be characterized using such identities.


It is worth remarking that for the Boolean domain the situation is
considerably simplified due to Post's classification of Boolean
$\pol(\Gamma)$~\cite{pos41}, and a large range of such problems have been
proven to admit dichotomies~\cite{creignou2008b}. For
example, Schaefers dichotomy theorem for $\SAT(\Gamma)$~\cite{sch78}
can be proven in an extremely straightforward manner using this approach.
However, for our purposes the above methods are too coarse-grained,
since the precise running time $O^*(c^n)$ for a problem $\SAT(\Gamma)$
is not preserved by the introduction of existentially quantified variables.
Hence, we are in need of more fine-grained algebraic tools than usual,
which can be applied as follows.

A \emph{partial operation} over $D$ (of some arity $r$)
is an operation $p: X \to D$ for some domain $X \subseteq D^r$.
Similar to the total case we again extend it to a partial operation on tuples over $D$:
for $x_1, \ldots, x_r \in D^n$,
we let $p(x_1, \ldots, x_r)[i]=p(x_1[i], \ldots, x_r[i])$
if this is defined for every position $i \in [n]$;
otherwise $p(x_1, \ldots, x_r)$ is undefined. 
Then $p$ is a \emph{partial polymorphism} of a relation $R \subseteq D^n$
if, for any $x_1, \ldots, x_r \in R$ such that $p(x_1,\ldots,x_r)$ is defined
we have $p(x_1, \ldots, x_r) \in R$. We will occasionally also say
that $R$ is {\em invariant} under the partial operation $p$.
A \emph{partial projection} is a subfunction of a projection;
such an operation preserves every possible relation. 
A partial polymorphism of a constraint language $\Gamma$
is a partial polymorphism of every relation $R \in \Gamma$ and we let
$\ppol(\Gamma)$ denote the set of all partial polymorphisms of
$\Gamma$. Similarly, given a set of partial operations $P$ we write
$\inv(P)$ to denote the set of relations invariant under $P$, and if
$P = \{p\}$ is singleton we write $\inv(p)$ instead of $\inv(\{p\})$.
Dually to this relaxed notion of a polymorphism, we have a strengthened
notion on the relational side: a {\em
quantifier-free primitive positive definition} (qfpp-definition) over
$\Gamma$ is a pp-definition without existential quantification. We
let $\pcclone{\Gamma}$ denote the smallest set of relations containing
$\Gamma$ and which is closed under qfpp-definitions, and
then obtain the following correspondence.

\begin{theorem}[\cite{Gei68,romov1981}]
  \label{theorem:pgalois_intro}
    $\Gamma \subseteq \pcclone{\Delta}$ if and only if  $\ppol(\Delta) \subseteq
    \ppol(\Gamma)$ for any constraint languages $\Gamma$ and $\Delta$.
\end{theorem}

With the help of this correspondence Jonsson et al.~\cite{jonsson2017}
proved that partial polymorphisms indeed can be used for studying the
fine-grained complexity of SAT and CSP.

\begin{theorem}
  Let $\Gamma$ and $\Delta$ be two finite constraint languages. If
  $\ppol(\Gamma) \subseteq \ppol(\Delta)$ then there is a
  polynomial-time many-one reduction from $\CSP(\Delta)$ to
  $\CSP(\Gamma)$ which does not increase the number of variables.
\end{theorem}

Unfortunately, this theorem is difficult to apply in practice since
it requires a good understanding of the structure of the closed sets
$\ppol(\Gamma)$ for all possible choices of $\Gamma$. Despite advances
made by several different researchers~\cite{lagerkvist2017,DBLP:conf/ismvl/CouceiroHSW14,Lagerkvist2016c,scholzel2014}, no such
classification is known even for Boolean $\Gamma$, and even less is
known for $\Gamma$ such that $\SAT(\Gamma)$ is NP-hard.  Hence, we
propose a method inspired by the rich algebraic toolbox developed for
studying the classical complexity of CSP: does the SETH-hardness of
$\SAT(\Gamma)$ and $\CSP(\Gamma)$ only depend on the identities
satisfied by the partial polymorphisms of $\Gamma$? On the one hand,
it is easily verified that if the only partial polymorphisms of
$\Gamma$ are the partial projections, then $\Gamma$ can qfpp-define
all $k$-clauses for every $k \geq 1$, and $\SAT(\Gamma)$ is SETH-hard.
On the other hand, we would have to show that every non-trivial
partial polymorphism $p$ allows the design of an algorithm that solves
$\SAT(\Gamma)$ in $O^*(c^n)$ time for some $c<2$.

One issue which speaks against the feasibility of this approach is
that individual partial polymorphisms are very weak restrictions.  For
one thing, it is known that for every finite set $P$ of partial
operations (that does not imply any non-trivial total operation), the
set $\inv(P)$ of all relations that are invariant under $P$ contains a
double-exponential number of relations as a function of the arity
$n$~\cite[Lemma~35]{lagerkvist2017c}.  Note that for a finite language
such as $k$-SAT, there are in contrast only $2^{O(n^k)}$ distinct
instances on $n$ variables.  Hence, languages $\inv(p)$ for a
single partial operation $p$ would be much richer than previously
studied problems.  Very similarly, in a related
study~\cite{LagerkvistW17CP}, it was shown that the existence of
so-called \emph{polynomial kernels} for $\SAT(\Gamma)$ cannot be
characterised by such a finite set $P$, whereas every finite problem,
as well as \textsc{Exact SAT} and problems defined via bounded-degree
polynomials, have polynomial kernels~\cite{JansenP16MFCS}.

Nevertheless, contrary to these earlier results, we will prove that
the presence of certain individual partial polymorphisms can be used
to design improved algorithms for SAT problems. As a starting point we
in the first hand consider the partial analogues of well-studied
polymorphisms resulting in tractable CSPs. For example, a {\em Maltsev
operation} is a ternary operation $\malt$ satisfying the two
identities $\malt(x,x,y) = y$ and $\malt(y,x,x) = y$, and is
well-known to result in tractable CSPs due to the algorithm by Bulatov
and Dalmau~\cite{bulatov2006b}. We may then define the {\em partial
  Maltsev operation} over a domain $D$ as the unique partial operation
which for all $x,y \in D$ satisfies these two identities, but which is
undefined otherwise. Similarly, it is possible to define partial variants of
$k$-ary {\em near unanimity} ($k$-NU) and $k$-ary {\em
  edge} ($k$-edge)
operations. These classes of operations are formally defined in
Section~\ref{section:patterns} and at the moment we will simply regard them as
well-behaved operations resulting in tractable CSPs, but we remark
that a 2-edge operation  is equivalent to a Maltsev operation and that
a ternary NU-operation is nothing else than a majority operation. It may also be
interesting to observe that 
the partial operations defined in this manner are unique for every
fixed domain, even though there may exist a large number of total
operations satisfying the identities.


\subsection{Our results and structure of the paper}
\label{sec:results}
For a partial polymorphism $p$, let \textsc{$\inv(p)$-SAT} refer
to the problem $\SAT(\Gamma)$ where $\Gamma=\inv(p)$. Hence, in this
problem every involved relation is invariant under the given partial
operation $p$. We will sometimes also refer to the CSP-variants of
these problems and denote these by $\inv(p)$-CSP (and tacitly assume
that the domain of the operation $p$ is clear from the context, or is not relevant).
We look at three related aspects of the complexity of these problems.
Let us first discuss our model more carefully.

\paragraph{Our questions and model.}
Since $\Gamma$ is infinite we first need to
fix a constraint representation. 
Let $R \subseteq \{0,1\}^r$ be a relation. 
An \emph{explicit representation} of $R$
is a list of all tuples $t \in R$. For infinite languages the explicit
representation is not always the most natural one since a relation may
contain exponentially many tuples with respect to the arity. This is
particuraly troublesome when proving lower bounds for
\textsc{$\inv(p)$-SAT} since we may not be able to construct relations
of arbitrary arity in the required time bound. Hence, we also consider
an {\em implicit representation}. In this model of representation a
contraint $R(x_1, \ldots, x_r)$ is represented by an oracle consisting
of a computable function which, given an assignment to variables $X
\subseteq \{x_1, \ldots, x_r\}$, can determine if this assignment can
be extended to an assignment to $\{x_1, \ldots, x_r\}$ consistent with
$R$.

\begin{example}
  For each $r \geq 3$ consider the relation $R^{r} = \{(x_1, \ldots,
  x_r) \in \{0,1\}^r \mid x_1 + \ldots + x_r$ is even$\}$. Even though
  $|R^r|$ is exponential with respect to $r$ it is not difficult to
  see that constraints over $R^r$ can be implicitly represented by
  computing the parity of the given assignment.
\end{example}

Given these definitions, we consider the following three notions of
improved algorithms.

\begin{definition}
  Let $\Gamma$ be an infinite constraint language. 
  \begin{enumerate}
  \item $\SAT(\Gamma)$ admits a \emph{non-uniform improved algorithm}
    with running time $O^*(c^n)$, $c<2$,
    if for every finite $\Gamma' \subset \Gamma$
    the problem $\SAT(\Gamma')$ can
    be solved in $O^*(c^n)$ time. 
  \item $\SAT(\Gamma)$ admits an improved algorithm
    \emph{in explicit representation} if $\SAT(\Gamma)$
    admits an improved algorithm for the problem variant
    where every relation is provided in explicit representation. 
  \item $\SAT(\Gamma)$ admits an improved algorithm
    \emph{in the oracle model} if $\SAT(\Gamma)$ admits
    an improved algorithm when constraints are provided
    only as extension oracles. 
  \end{enumerate}
\end{definition}

Note that for a non-uniform improved algorithm,
the representation does not matter. Also note that
these are gradually stronger requirements,
and that in these terms, SETH states that
CNF-SAT does not admit even a non-uniform
improved algorithm. On the other hand,
allowing constraints of unbounded arity
via oracle access can be useful; for example,
the $n$-ary constraint ($\sum_{i=1}^n x_i = k$)
has a simple extension oracle, and if included
in the language, can be used to phrase
optimisation problems as oracle-access SAT problems. 

To restrict our scope, we focus on constraint languages that are
closed under variable negation. Informally, this means that whenever
$R \in \Gamma$, in addition to constraints $R(x_1, \ldots, x_r)$ on
only positive variables, we are also allowed to impose constraints
such as $R(x_1, \ldots, \neg x_i, \ldots, x_r)$ with some occurrences
of variables $x_i$ negated in the constraint.  More formally, it means
that for every $R \in \Gamma$, and for every subset $S \subseteq
[\ar(R)]$ of positions of $R$, the relation produced by negating every
tuple $t \in R$ in positions $S$ is also contained in $\Gamma$.  In
this case, we say that $\Gamma$ is \emph{sign-symmetric}.  This is a
natural restriction which holds for many well-studied constraint
language, e.g., the languges corresponding to $k$-SAT,
\textsc{1-in-$k$-SAT} and the roots of bounded-degree polynomials are
all sign-symmetric.  Furthermore, it is known that the expressive
power of a sign-symmetric constraint language is characterised by a
restricted kind of partial polymorphism which we refer to as
\emph{pSDI-operations} (for partial, self-dual and
idempotent)~\cite{lagerkvist2016,lagerkvist2015}.  Thus, the
restriction to sign-symmetric languages corresponds directly to a
restriction on the algebraic level. Most importantly, the Boolean partial
operations arising from system of identities of the form considered in
Section~\ref{sec:algebraic_approach} are guaranteed to be pSDI.

\paragraph{The fine-grained structure of NP-hard SAT problems.}
The first part of the paper, Section~\ref{section:structure}, is dedicated to
explaining the the structure of pSDI-operations. Due to the algebraic
correspondence between partial polymorphisms and qfpp-definability
this also serves as a classification of the NP-hard SAT problems we
need to consider for constructing improved algorithms.

First, we study the structure of single pSDI-operations $p$ that
impose some non-trivial restrictions on the expressive power of
$\Gamma$.  We particularly consider the weakest such operations, i.e.,
such that the language $\Gamma=\inv(p)$ is as rich as possible.  In
particular, we consider $p$ such that every subfunction of $p$ which
is pSDI is
a partial projection.  Let us
refer to such an operation as being \emph{minimal}. For example, the partial variants of Maltsev,
$k$-NU, and $k$-edge operations are all minimal.  Equipped with this
notion we then show that minimal
pSDI-operations are naturally organised into \emph{levels}, with a
structure as follows.
\begin{itemize}
\item There is a single minimal operation on level 2, which is
  the partial Maltsev, or, equivalently, the partial 2-edge operation. 
  This is also equivalent to the 2-universal operation defined below.
\item For every other minimal pSDI-operation $p$,
  there is a unique largest constant $k$ such that
  $p$ is implied by the partial $k$-NU
  operation $\near_k$. We refer to this as the \emph{level} of $k$.
  Thus, the partial $k$-NU operation is the strongest operation on level
  $k \geq 3$.
\item For every level $k \geq 2$, there is also
  a unique weakest pSDI-operation $\universal_k$ which we refer to as the \emph{$k$-universal} operation,
  such that $\universal_k$ is implied by every operation on level $k$. 
\item The language $\sat{k}$ corresponding to $k$-SAT
  is preserved by the partial $(k+1)$-NU operation,
  but not by any operation on a previous level; 
  and every sign-symmetric language $\Gamma$ that is not preserved
  by the $k$-universal operation can qfpp-define $\sat{k}$.
\item Finally, as an interesting case, roots of polynomials 
  of degree at most $d$ are preserved by the $(d+1)$-universal
  operation, but not by any other operation on a level up to $d+1$.
\end{itemize}
Thus, the levels of minimal pSDI-operations correspond to a natural
notion of difficulty.  It also follows that if a sign-symmetric
language $\Gamma$ is not preserved by the $k$-universal operation for
any constant $k$, then $\SAT(\Gamma)$ is trivially SETH-hard, whereas
every other language $\Gamma$ has some kind of restriction on its
expressive power.  We also note that there is no known case of a
problem known to be SETH-hard, which fits into a framework of
searching through the set $\{0,1\}^n$ for a solution, and which is
$k$-universal for any $k$.  Hence, it is consistent with our present
knowledge that every $k$-universal problem $\SAT(\Gamma)$ admits an
improved algorithm.

Last, we remark that although we in this paper are mainly interested
in the time complexity of SAT, the classification of minimal
pSDI-operations in this section may be of independant interest for any
Boolean problem compatible with qfpp-definitions. 
In this vein, we also give a ``vertical'' result in the above
hiearchy, and show that every sign-symmetric constraint language
$\Gamma$ not preserved by the partial $k$-NU operation for any $k$ 
can qfpp-define either 1-in-$k$-clauses of all arities, 
or counting constraints modulo $p$ of all arities for some fixed
prime $p$.  This result is the main technical challenge in this
section, and relies on an application of Szemer\'edi's
theorem~\cite{SzemerediThm} to analyse the structure of
symmetric relations $R \notin \inv(\near_k)$. 

\paragraph{Upper and Lower Bounds on the SAT problem.}
Second, in Section~\ref{section:upper} and Section~\ref{section:lower}, we consider the strength of the problem $\inv(p)$-SAT
for various pSDI-operations $p$, with an interest in bounding 
the value $c(\Gamma)$ for $\Gamma=\inv(p)$ from above and below.
The first question here is the matter of constraint representation. As
mentioned previously, the language $\inv(p)$ contains a double-exponential
number of relations of arity $r$ as a function of $r$; hence any fixed
representation would in the worst case use $2^{O(r)}$ bits just to encode
the relations. This becomes an issue when we allow constraints of
unbounded arity. Recall that we consider three alternatives for
representation: explicit representation, extension oracles, and
non-uniform algorithms where the particular choice of representation
does not matter. We then obtain the following results.





\begin{itemize}
\item When $p$ is the partial 2-edge operation, we refer to
  $\inv(p)$-$\SAT$ as \textsc{2-edge-SAT}.
  We show that \textsc{2-edge-SAT} can be solved in $O^*(2^{\frac{n}{2}})$ time
  in the oracle setting using a meet-in-the-middle strategy
  combined with the computation of a kind of canonical labels
  for partial assignments, similarly to the $O^*(2^{\frac{n}{2}})$-time
  algorithm for \textsc{Subset Sum} with $n$ integers~\cite{horowitz1974}.
  A similar improved algorithm is possible for 
  the generalisation to \textsc{2-edge-CSP}, i.e., for fixed
  non-Boolean domains. 
  Furthermore, if $c(\inv(p)) < 2^{1/2}$ in the extension oracle
  setting, then \textsc{Subset Sum} 
  can be solved in $O^*(2^{(\frac{1}{2}-\varepsilon)n})$ for some $\varepsilon > 0$,
  which is a long-standing open problem. 
\item When $p$ is the partial $k$-NU operation, we refer to $\inv(p)$-SAT
  as \textsc{$k$-NU-SAT}. For $k=3$, this problem is equivalent
  to 2-SAT, and hence in P, but the generalisation \textsc{3-NU-CSP}
  to larger fixed domains is NP-hard and
  admits an improved algorithm 
  using fast matrix multiplication,
  similarly to the well-known algorithm for the CSP problem over
  binary constraints.
\item For $k>3$, we show two conditional connections. First, if 
  the \textsc{$(k, k-1)$-hyperclique} problem for hypergraphs
  with ground set of size $n$ can be solved in time $O(n^{k-\varepsilon})$
  for any $\varepsilon > 0$, then both \textsc{$k$-NU-SAT} and 
  \textsc{$k$-NU-CSP} admit improved algorithms in the oracle
  setting. 
  Second, if the Erd\H{o}s-Rado \emph{sunflower conjecture}~\cite{JLMS:JLMS0085} holds for 
  sunflowers with $k$ sets, then $k$-NU-SAT admits an improved
  algorithm via a local search strategy in the explicit representation,
  similar to Sch\"oning's algorithm for $k$-SAT~\cite{schoning1999}. 
\item We also investigate the case that $p$ is the partial 3-edge
  operation $\edge_3$, and give a partial result. Assume that 
  every relation $R$  in the input is either preserved by the partial
  2-edge relation, or by $\near_3$, or $R$ is \emph{symmetric}
  and preserved by $\edge_3$ -- i.e., whether $t \in R$ depends only
  on the Hamming weight of $t$. 
  Then the SAT problem has an improved
  algorithm via a reduction to the problem of 
  finding monochromatic triangles in an edge-coloured graph, 
  which in turn can be solved using fast algorithms for 
  triangle finding in sparse graphs. 
  We do not know whether this strategy generalises to non-symmetric relations.
\end{itemize}
For further classes, we note that $\SAT(\Gamma)$ contains some highly
challenging special cases. In particular an algorithm for the 
$k$-universal languages for $k>2$ would need to generalise the
algorithm of Lokshtanov et al.\ for bounded-degree
polynomials~\cite{LokshtanovPTWY17SODA}, while only
using the abstract properties guaranteed by $\universal_k$.

Finally, we show lower bounds in the oracle extension model: for every
minimal pSDI-operation $p$, we get a concrete lower 
bound $c(\inv(p)) \geq c_k > 1$ assuming the randomized SETH,  
where $k$ is the level of $p$. That is, unless SETH is false, no
algorithm can solve $\inv(p)$-$\SAT$ in time
$O^*(c_k^{(1-\varepsilon)n})$ for any $\varepsilon > 0$ 
and any $p$ at level $k$. The bound $c_k$ converges to 2 at a rate of
$2-c_k = \Theta(\frac{\log k}{k})$. 

\paragraph{A connection to polynomial-time problems.}
Finally, we make some connections between the
$\inv(p)$-SAT and $\inv(p)$-CSP problems and some problems in
polynomial-time algorithms. We show that the minimal pSDI-operations
generalise not only to CSP problems on fixed domains, 
but to abstract conditions on ``CSP-like'' problems on a domain of 
size $n$ and with $d=\Theta(1)$ variables. We refer to this as
the \emph{abstract $\inv(p)$-problem}. Any solution to such a problem
that runs in time $O(n^{d-\varepsilon})$ for any $\varepsilon > 0$
implies an improved algorithm for the corresponding $\inv(p)$-SAT 
and $\inv(p)$-CSP problems in the oracle setting for every fixed domain.
This lies behind the improved algorithms for \textsc{2-edge-CSP}
and \textsc{3-NU-CSP}. 

However, there is some indication that these problems may be tougher
than the original problems, 
since the reduction loses a significant amount of instance structure
(e.g., the local search strategy for $k$-NU-$\SAT$ cannot be lifted to
the abstract problem). 
In fact, there are conjectures that would prevent improved algorithms
for most cases of the abstract problem considered in this article:
\begin{itemize}
\item The abstract $k$-NU problem is equivalent to
  \textsc{$(k,k-1)$-hyperclique}, i.e., the problem of finding a
  $k$-hyperclique in a $(k-1)$-regular hypergraph.
  Thus, it has an improved algorithm for $k=3$ but the status is unknown 
  for $k>3$. Moreover, the general $(l,k)$-hyperclique problem for $l
  > k$ has been conjectured to require $n^{l - o(1)}$ time~\cite{williams2018}.
\item The abstract 3-universal problem contains the problem of finding
  a zero-weight triangle in an edge-weighted graph with arbitrary edge weights.
  This does not admit an improved algorithm unless the 3-SUM conjecture fails
  (but SETH-hardness is not known)~\cite{VWilliamsW13SICOMP}.
\end{itemize}
Considering the connections, we still consider it useful to ask
which minimal pSDI-operations $p$ suffice to guarantee an improved
algorithm for the abstract $\inv(p)$-problem. We leave this question
for future work. 

\subsection{Technical notes and proof methods}

Let us now give a few more details about the proofs of the above
results.

The structural characterisation builds on a description of minimal
non-trivial pSDI-operations (Lemma~\ref{lemma:whichminimal}) --- 
they are precisely the operations produced by padding the partial
$k$-NU operation by additional arguments. The weakest and strongest
operations on each level follow from this almost by definition. It
also follows that the operations on each level $k$ are characterized 
by the presence or absence of each of roughly $2^k$ possible types of
padding argument. 
Note that such a padding makes an operation weaker; e.g., in order to
apply the partial majority operation to a sequence of tuples 
$t_1, \ldots, t_k \in R$ for some relation $R$, in a padded version of
arity $r$ we require that $R$ further contains a sequence of tuples
$t_{k+1}$, \ldots, $t_r$ determined by the padding arguments from
the tuples $t_1, \ldots, t_k$. 

This also provides a way to think about
the consequences of not being preserved by such an operation. 
Assume e.g.\ that a relation $R$ is not preserved by $\near_k$.
Then by definition there are $t_1, \ldots, t_k \in R$
such that $\near_k(t_1,\ldots,t_k) = t$ is defined, and by
sign-symmetry we may assume that $t$ is the constant 0-tuple $0^{\ar(R)}$. 
Then the witness produces a partition of the arguments of $R$,
in a way which can be used to implement a relation $R'$ 
of arity $k$ which accepts every tuple of weight 1 but none of weight
0. However, we have no information at this point about the remaining
tuples in $R'$. 
Continuing this line of reasoning to derive a consequence for an
infinite sign-symmetric language $\Gamma$ with $\near_k \notin
\ppol(\Gamma)$ for every $k$, we first observe that 
we can define a symmetric relation $R'' \notin \inv(\near_k)$
as a conjunction of $k!$ applications of $R'$ under argument
permutation, then (as announced) analyse the possibilities for
families of such relations using Szemer\'edi's theorem.
In particular, a broken arithmetic progression of $i$ accepted weights 
in such a relation implies that we can qfpp-define an
$i+1$-clause using $R$. 

By contrast, if $\universal_k \notin \ppol(R)$, then the tuples 
$t_1, \ldots, t_{2^k-1} \in R$ required by the arguments of
$\universal_k$ imply that such a relation $R'$ must 
have $|R'|=2^k-1$, i.e., it must be the relation corresponding to a
$k$-clause. 

Moving on to the algorithmic applications, most of the positive
results are relatively straight-forward applications of known ideas; 
the interesting aspect is that the applicability of these ideas
follows from such simple conditions as the minimal pSDI-operations.
Here, we particularly wish to highlight the conjectural connection to
local search.  Recall that Sch\"oning's algorithm~\cite{schoning1999}
reduces $k$-SAT to several applications of local search, i.e.,
given a starting point $x \in \{0,1\}^n$ and a parameter $t$, 
find a satisfying assignment within Hamming distance $t$ of $x$. 
By sign-symmetry, for our problem this reduces to the case $x=0^n$ 
(alternatively, one could use monotone local search; cf. Fomin et
al.~\cite{FominGLS16localsearch}). 
Now, consider the set of all minimal tuples in any relation
$R \in \inv(\near_k)$ with $0^{\ar(R)} \notin R$. 
It is easy to see that by the $\near_k$-condition,  this set
does not contain a sunflower of $k$ sets, and by the sunflower
conjecture, this implies that for every $i$ there are at most $C^i$
such minimal tuples in $R$ of weight $i$ for some $C$. A simple computation shows
that a recursive algorithm that finds an unsatisfied relation $R$, 
enumerates minimal tuples in it, and recursively proceeds 
from every such tuple yields a total searching time of $2^{O(t)}$, 
which would be precisely sufficient to yield an improved algorithm for
\textsc{$k$-NU-SAT}. This algorithm uses the explicit representation
in order to be able to enumerate such minimal tuples. It is
an interesting open question whether this can be achieved efficiently
in the oracle setting.


Finally, we move on to our lower bounds. These are of two kinds, a
reduction from \textsc{Subset Sum} to \textsc{2-edge-SAT}, 
and the generic lower bound under SETH against any problem
\textsc{$\inv(p)$-SAT}. 
For the former, recall that the partial 2-edge operation is equivalent
to $\universal_2$, and thus contains all constraints which can be
phrased as linear equations, e.g., \textsc{Subset Sum} instances.
But we are also required to provide an extension oracle for each
constraint, which is clearly infeasible if we plug in the
\textsc{Subset Sum} equation as-is. 
However, this is easily solved by splitting the binary expansion of
the target number into $O(\sqrt{n})$ blocks of $O(\sqrt{n})$ bits each.
With some moderate guessing, each block reduces to one linear
equation, and via the tabulation algorithm for \textsc{Subset Sum} 
an extension oracle each such block can be produced with
a query time of $2^{O(\sqrt{n})}$. 

The generic bounds, in turn, work via a generic padding argument: we
show that for every level $k$, and any set $X$ of $n$ variables, there
is a universal padding formula $R(X,Y)$ on $|Y|=\Theta(n)$ additional
variables such that $R'(X,Y) \equiv R'(X) \land R(X,Y)$ is $k$-NU for
any relation $R'(X)$. Furthermore, random parity-check variables
suffice to produce this padding formula, allowing for an efficient
extension oracle for the relation $R'(X,Y)$. Finally, by the
regularity of the padding formula, we can reuse the same variables $Y$
for all constraints in an input instance of $q$-SAT, for any $q$,
and only pay with $|Y|=\Theta(n)$ extra variables in total. 

The fact that some such padding exists was previously
known~\cite{lagerkvist2017c}. Recall that every operation $p$
considered has at least one tuple of values for which it is
undefined. Then, if we add enough random variables, for every
attempt  $p(t_1, \ldots, t_r)$ of finding a valid application of $p$
on a relation $R$ there will be a padding variable $j$ such
that $(t_1[j], \ldots, t_r[j])$ takes the values of such a tuple, 
and $p$ is undefined. The fact that parity-check variables suffice
in our case follows from the fact that $p$ contains $k$ arguments that
form a partial $k$-NU operation. It is easy to check that almost all
parity-check variables form an undefined tuple of values already over
these arguments. 
This construction could be derandomized using a universal hash family,
possibly at the cost of a larger constant $|Y|/|X|$, but we do not
pursue this.

\subsection{Related work}
Our work can be seen as an amalgamation of the following areas:
fine-grained time complexity and lower bounds under the SETH, and
the algebraic approach for studying classical complexity of CSP. 

Concerning the former, SETH has turned out to be a highly useful
conjecture for exact algorithms since a relative lower bound from SETH
shows that any further improvements also implies a breakthrough
speed-up for SAT. Many different problems have been shown
to admit lower bounds via the SETH, but in the current context of SAT,
in addition to the foundational works of Impagliazzo et al.~\cite{impagliazzo2001,impagliazzo98}
it is worth mentioning the lower bound for \textsc{Not-all-equal SAT}
(NAE-SAT) by Cygan et al.~\cite{CyganDLMNOPSW16TALG} and the lower bound for $\Pi_23$-SAT
by Calabro et al.~\cite{Calabro2013}. However, to the best of our knowledge,
all concrete lower bounds using SETH for exponential-time algorithms
falls into one of the following cases: 
either
  the lower bound matches the running time of a trivial algorithm, as in
  the case of \textsc{Hitting Set}, \textsc{NAE-SAT}, and
  $\Pi_2$-3-SAT, showing that no improvement is possible;
  or the lower bounds are with respect to a much more permissive
  complexity parameter than $n$, such as
  treewidth~\cite{Lokshtanov:2011:KAG:2133036.2133097}. 
The one other example we are aware of is from the study of 
infinite-domain CSPs by Jonsson and
Lagerkvist~\cite{Jonsson:Lagerkvist:ai2017},
who obtained upper bounds of the form $O^{*}(2^{f(n)})$ for non-linear
functions $f$  and a lower bound stating that the CSPs are not
solvable in $O(c^n)$ time for any constant $c$. These bounds are therefore in a sense closer to non-subexponentiality results usually obtained from the ETH.
SETH and other conjectures have also seen significant applications
over recent years in producing conditional lower bounds for
polynomial-time solvable problems, but these are only tangentially
relevant here.



With regards to the algebraic approach we wish to highlight a few
related but different results. Partial polymorphisms and the link to
qfpp-definitions were first introduced to the CSP community by Schnoor
\& Schnoor~\cite{schnoor2008a} even though these notions were
well-known in the algebraic community much
longer~\cite{Gei68,romov1981}. However, the principal motivation by
Schnoor \& Schnoor was to obtain dichotomy theorems for CSP-like
problems incompatible with existential quantification, and the
explicit connection to fine-grained time complexity of CSP was not
realized until later by Jonsson et al.~\cite{jonsson2017}. This work
utilized a {\em lattice-informed} approach which exploited the
structure of the inclusion structure of closed sets of partial
polymorphisms, in order to identify an NP-complete $\SAT(\Gamma)$
problem such that $c(\Gamma) \leq c(\Delta)$ for every other
NP-complete $\SAT(\Delta)$. This problem was referred to as the {\em
easiest NP-complete SAT problem} and was later generalized to a broad
class of finite-domain CSPs~\cite{lagerkvist2017e}. However,
continued advancements in understanding this inclusion structure
revealed that even severely restricted classes of constraint languages
had a very complicated
structure~\cite{lagerkvist2017,Lagerkvist2016c}. In a similar vein of
negative results it was also proven that (1) $\ppol(\Gamma)$ cannot be
generated by any finite set of partial operations whenever $\Gamma$ is
finite and SAT$(\Gamma)$ is NP-hard, and (2) if $P$ is a finite set of
partial operations such that $\inv(P)$-SAT is NP-hard, then any
pp-definable relation over $\inv(P)$ can be transformed into a
pp-definition using only a linear number of existentially quantified
variables~\cite{lagerkvist2017c}. In plain language, these results show
that finite constraint languages result in complex partial
polymorphisms, and that simple partial polymorphisms result in complex
constraint languages.
A previous attempt at grappling with this difficulty provided closure
operators that generate $\ppol(\Gamma)$ for a finite $\Gamma$ from a
finite basis~\cite{lagerkvist2016}, but this intrinsically uses that
$\Gamma$ is finite, and is not applicable in the current paper. 


Our approach in this paper avoids the pitfalls of the lattice-informed
approach since it is sufficient to understand the behaviour of
individual pSDI-operations.  This is in line with how the research
programme of classifying the complexity of finite-domain CSPs evolved
into a project of describing properties of operations defined by
system of identities (see the survey by Barto et al.\ for more
details~\cite{barto2017}).

Another related paper by the present authors investigates the
existence of \emph{polynomial} (or linear) \emph{kernels} for problems
$\SAT(\Gamma)$, using ideas of extending the language $\Gamma$ into a
tractable CSP on a larger domain~\cite{LagerkvistW17CP},
including extensions into 2-edge (i.e., Maltsev) and $k$-edge languages. 
However, there is no concrete technical connection between that paper
and this one, as having polynomial kernels turns out to be a much more
restricted property than admitting improved algorithms.

\subsection{Concluding remarks and open questions}
\label{sec:conclusions}  
Our principal motivation in this paper is to study the SETH-hardness
of the parameterized SAT$(\Gamma)$ problem. To simplify our study we
restricted our focus to sign-symmetric constraint languages, which 
is a common assumption for SAT problems studied in practice. Moreover,
due to the connection between sign-symmetric constraint languages and
pSDI-operations, understanding the inclusion structure between
sign-symmetric constraint languages is tantamount to describing the
expressive power of pSDI-operations. Even better, pSDI-operations can
in many cases be understood as the partial analouges of well-studied
operations such as Maltsev operations, NU-operations and
edge-operations, making them easier to reason with.

The main open question is whether our results can be strengthened into
a dichotomy for sign-symmetric SAT problems. One direction is already clear: if
$\Gamma$ is not preserved by any $k$-universal operation then
SAT$(\Gamma)$ is SETH-hard and does not admit an improved algorithm
without breaking the SETH. The other direction is harder and
requires a substantially better understanding of languages invariant
under a given $k$-universal operation; such languages include, but
are not limited to, relations expressible as roots of polynomial equations
of degree at most $k+1$, where an improved algorithm is
known~\cite{LokshtanovPTWY17SODA}. It is not clear at this point
how much richer the set $\inv(\universal_k)$ is, compared to this
class of problems. 
Existing (conjectured) lower bounds against polynomial-time problems
captured by abstract $\inv(p)$-problems also indicate that 
the problem might be more difficult for remaining cases.
We also proved that the SAT problems under consideration admit lower
bounds under the SETH. To the best of our knowledge, this is the first
result showcasing both a non-trivial upper bound and a concrete lower
bound under the SETH in terms of a natural parameter $n$. 
These bounds were obtained in the extension
oracle setting and it is currently unclear if matching bounds can also
be obtained if constraints are represented explicitly. The padding
construction is still valid in this setting, but it is a challenge 
to apply it without creating constraints with exponentially many tuples.

Last, our approach easily extends to finite-domain CSPs, as evidenced
by the improved algorithms for 2-edge-CSP and 3-NU-CSP. The
notion of a pSDI-operation is only relevant in the Boolean domain, but
a similar notion can likely be defined for arbitrary finite
domains. For example, instead of self-duality, essentially meaning
that the partial operation is closed under negation, we would require
that the operation is closed under every unary operation over the
domain. However, it is not clear if the inclusion structure of such
generalized pSDI-operations can be characterized in a similar
hierarchy as the Boolean pSDI-operations.




\section{Preliminaries}
\label{sec:prel}

A $k$-ary {\em relation} over a domain $D$ is a subset of $D^k$. If
$t = (x_1, \ldots, x_n)$ is a $k$-ary tuple we for every $1 \leq i
\leq k$ let $t[i] = x_i$, and if $i_1, \ldots, i_{k'} \in [k] = \{1,
\ldots, k\}$ we write $\pro_{i_1, \ldots, i_{k'}}(t) = (t[i_1],
\ldots, t[i_{k'}])$ for the {\em projection} of $t$ on the coordinates
$i_1, \ldots, i_{k'}$. This notation easily extends to relations and
we write $\pro_{i_1, \ldots, i_{k'}}(R)$ for the relation
$\{\pro_{i_1, \ldots, i_{k'}}(t) \mid t \in R\}$.

A set
of relations is called a {\em constraint language}, or simply a {\em
language}, and will usually be denoted by $\Gamma$ and $\Delta$.  We
will typically define relations either by their defining logical
formulas or by their defining equations. For example, the relation
$R_{1/3} = \{(0,0,1), (0,1,0), (1,0,0)\}$ may be defined by the
expression $R_{1/3} \equiv x_1 + x_2 + x_3 = 1 $. However, we will not
always make a sharp distinction between relations and their defining
logical formulas and will sometimes treat e.g.\ a $k$-clause as a relation. We write $\ar(R)$
for the arity of a relation $R$, and use the notation $\eq_D$ to
denote the equality relation $\{(x,x) \mid x \in D\}$ over $D$.

A $k$-ary relation $R$ is said to be {\em totally symmetric}, or just
{\em symmetric}, if there exists a set $S \subseteq [k] = \{1, \ldots,
k\}$ such that $(x_1, \ldots, x_k) \in R$ if and only if $x_1 + \ldots
+ x_k \in S$. For example, $R_{1/3}$ is totally symmetric as witnessed
by the set $S = \{1\}$. Symmetric relations will prove to be useful
since it is sometimes considerably simpler to describe the symmetric
relations invariant under a partial operation. 

\subsection{The parameterized SAT and CSP Problems}

Let $\Gamma$ be a Boolean constraint language. The {\em parameterized
  satisfiability problem} over $\Gamma$ ($\SAT(\Gamma)$) is the
computational decision problem defined as follows.

\smallskip

\noindent
{\sc Instance:} A set $V$ of variables and a set $C$ of constraint
applications $R(v_1,\ldots,v_{k})$ where $R \in \Gamma$, $\ar(R) = k$,
and $v_1,\ldots,v_{k} \in V$.

\noindent
{\sc Question:} Is there a function $f : V \rightarrow \{0,1\}$ such
that $(f(v_1),\ldots,f(v_{k})) \in R$ for each 
$R(v_1,\ldots,v_{k})$~in~$C$?

\smallskip

The {\em constraint satisfaction problem} over a constraint language
$\Gamma$ ($\CSP(\Gamma)$) is defined analogously with the only distinction that
$\Gamma$ is not necessarily Boolean. We write $(d,k)$-CSP for the CSP
problem over a domain with $d$ elements where each constraint has
arity at most $k$. 

\subsection{The extension oracle model}

Recall from Section~\ref{sec:results} that we consider two distinct
representations of $\SAT$ and $\CSP$ instances. We now define these
in more detail. In the first representation
each relation $R$ occurring in a constraint $R(x_1,
\ldots, x_k)$ is represented as a list of tuples. We call this
representation the {\em explicit representation}. This is one of the
most frequently occurring representation methods in the algebraic
approach to CSP, but it is fair to say that it is not convenient in
any practical application since a relation may contain exponentially many
tuples with respect to the number of arguments. We therefore
consider a more implicit representation where each constraint is
represented by a procedure which can verify whether
a partial assignment of its variables is consistent with the
constraint.

\begin{definition}
  Let $R$ be an $n$-ary relation over a set $D$. A computable function which given indices $i_1, \ldots, i_{n'} \in
  [n]$ and $t \in D^{n'}$ answers yes if and only if $t \in \pro_{i_1, \ldots, i_{n'}}(R)$
  is called an {\em extension oracle representation} of $R$.
\end{definition}

Hence, given a constraint $R(x_1, \ldots, x_n)$ and a partial truth
assignment $f : X \rightarrow D$, $X \subseteq \{x_1, \ldots, x_n\}$,
the extension oracle representation can be used to 
decide whether $f$ can be completed into a satisfying assignment
of $R(x_1, \ldots, x_n)$.

\begin{example}
  CNF-SAT can be succinctly represented in the extension oracle
  model. Consider e.g.\ a positive clause $(x_1 \vee \ldots \vee x_n)$ and a
  partial truth assignment $f$ on $\{x_1, \ldots, x_n\}$. We can then
  answer yes if and only if not every variable $x_i$ occurring in the
  clause is assigned the value 0.
\end{example}

\subsection{Sign-symmetric constraint languages}
\label{section:sign}
An {\em $n$-ary sign pattern} is an tuple $s$ where $s[i] \in
\{+,-\}$ for each $1 \leq i \leq n$. If $t$ is an $n$-ary Boolean
tuple and $s$ an $n$-ary sign pattern then we let $t^{s}$ be the tuple
where $t^s[i] = t[i]$ if $s[i] = +$ and $t^s[i] = 1 - t[i]$ if $s[i] =
-$. Similarly, if if $R$ is a Boolean relation
and $s$ an $n$-ary sign pattern we by $R^{s}$ denote the relation
$\{t^s \mid t \in R\}$. Last, for $1 \leq i \leq n$ and $c \in
\{0,1\}$ we let $R_{i=c} = \{t \mid t \in R, t[i] = c\}$ be the
relation resulting from freezing the $i$th argument of $R$ to $c$.

\begin{definition}
  A Boolean constraint language $\Gamma$ is said to be {\em
    sign-symmetric} if (1) $R^{s} \in \Gamma$ for every
  $n$-ary $R \in \Gamma$ and every $n$-ary sign pattern $s$ and
  (2) $R_{i = c} \in \Gamma$ for every $c \in \{0,1\}$ and every $1
  \leq i \leq n$.
\end{definition}


\subsection{Partial polymorphisms and quantifier-free primitive
  positive definitions}
Let $D$ be a finite set of values. A $k$-ary {\em partial operation},
or a {\em partial function}, $f$
over $D$ is a mapping $X \rightarrow D$ where $X \subseteq D^k$. The
set $X$ is said to be the {\em domain} of $f$ and we let $\domain(f) =
X$ denote this set and $\ar(f) = k$ denote the arity of $f$. If $f$
and $g$ are two $n$-ary partial operations over $D$ such that
$\domain(g) \subseteq \domain(f)$ and $g(x_1, \ldots, x_n) = f(x_1,
\ldots, x_n)$ for every $(x_1, \ldots, x_n) \in \domain(g)$ then $g$
is said to be a {\em subfunction} of $g$. For $n \geq 1$ the $i$-ary
{\em projection}, $1 \leq i \leq n$,
is the operation $\pi^n_{i}(x_1, \ldots, x_i, \ldots, x_n) = x_i$ and
a {\em partial projection} is any subfunction of a total projection.

If $R$ is an $n$-ary relation over $D$ and $f$ a $k$-ary partial
operation over $D$ we say that $f$ is a {\em partial polymorphism} of
$R$, that $R$ is {\em invariant} under $f$, or that $f$ {\em
  preserves} $R$, if $f(t_1, \ldots, t_k)
\in t$ or $f(t_1, \ldots, t_k)$ is undefined, for each sequence of
tuples $t_1, \ldots, t_k$. We let $\ppol(R)$ be the set of all partial
polymorphisms of the relation $R$, and if $\Gamma$ is a constraint
language we let $\ppol(\Gamma)$ denote the set of partial operations
preserving each relation in $\Gamma$. The notion of a total
polymorphism can be defined simply by requiring that $f$ is total,
i.e., $\domain(f) = D^k$, and we let $\pol(\Gamma)$ be the set of all
total polymorphsims of the constraint language $\Gamma$.  Similarly,
if $P$ is a set of partial operations we let $\inv(P)$ be the set of
all relations invariant under $P$. Each set of partial operations $P$
naturally induces a SAT problem $\SAT(\inv(P))$ where each relation
involved in a constraint is preserved by every partial operation in
$P$. Recall from Section~\ref{sec:results} that we as a shorthand denote this
problem by \textsc{$\inv(P)$-SAT}.
The two operators $\inv(\cdot)$ and
$\ppol(\cdot)$ are related by the following {\em Galois connection}.

\begin{theorem}[\cite{Gei68,romov1981}]
  \label{theorem:pgalois}
  Let $\Gamma$ and $\Delta$ be two constraint languages. Then 
  $\Gamma \subseteq \inv(\ppol(\Delta))$ if and only if  $\ppol(\Delta) \subseteq
    \ppol(\Gamma)$.
\end{theorem}

The applicability of partial polymorphism in the context of fine-grained
time complexity might not be evident from these definitions. However,
sets of the form $\inv(P)$, called {\em weak systems}
or {\em weak co-clones}, are closed under certain restricted
first-order formulas which are highly useful in this context. Say that
a $k$-ary relation $R$ has a {\em quantifier-free definition}
(qfpp-definition) over a constraint language $\Gamma$ over a domain $D$ if $R(x_1,
\ldots, x_k) \equiv R_1(\mathbf{x}_1) \land \ldots \land
R_m(\mathbf{x}_m)$ where each $R_i \in \Gamma \cup \{\eq_D\}$ and each
$\mathbf{x}_i$ is a tuple of variables of length $\ar(R_i)$. It is
then known that $\inv(P)$ for any set of partial operations $P$ is
closed under taking qfpp-definitions. With this property the following
theorem is then a straightforward consequence.

\begin{theorem}\cite{jonsson2017}
Let $\Gamma$ and $\Delta$ be two finite constraint languages. If
$\ppol(\Gamma) \subseteq \ppol(\Delta)$ then there exists a
polynomial-time many-one reduction from $\SAT(\Delta)$ to
$\SAT(\Gamma)$ which maps an instance $(V,C)$ of $\SAT(\Delta)$ to an
instance 
$(V',C')$ of $\SAT(\Gamma)$ where $|V'| \leq |V|$ and
$|C'| \leq c|C|$, 
where $c$ depends only on $\Gamma$ and $\Delta$.
\end{theorem}

In particular this implies that if $\CSP(\Gamma)$ is solvable in
$O(c^n)$ time and $\ppol(\Gamma) \subseteq \ppol(\Delta)$ then
$\CSP(\Delta)$ is solvable in $O(c^n)$ time, too. We will now briefly
describe the closure properties of $\ppol(\Gamma)$, which are usually
called {\em strong partial clones}. First, if $f, g_1, \ldots, g_m \in
\ppol(\Gamma)$ where $f$ is $m$-ary and each $g_i$ is $n$-ary, then
the {\em composition} $f \circ g_1, \ldots, g_{m}(x_1,
\ldots, x_n) = f(g_1(x_1, \ldots, x_n), \ldots, g_m(x_1, \ldots,
x_n))$ is also included in $\ppol(\Gamma)$. This operation will be
defined on a tuple $(x_1, \ldots, x_n) \in D^n$ if and only if each
$g_i(x_1, \ldots, x_n)$ is defined and the resulting application over
$f$ is defined. Second, $\ppol(\Gamma)$ contains every partial
projection, which is known to imply that $\ppol(\Gamma)$ is closed
under taking subfunctions (i.e., if $f \in \ppol(\Gamma)$ then every
subfunction of $f$ is included in $\ppol(\Gamma)$). If $P$ is a set of
partial operations we write $\strongof{P} = \ppol(\inv(P))$ for the
smallest strong partial clone containing $P$.

\subsection{Polymorphism patterns}
\label{section:patterns}
In this section we describe a method for constructing partial
polymorphisms that have a strong connection to the sign-symmetric
constraint languages defined in Section~\ref{section:sign}. As a shorthand we will sometimes denote the $k$-ary constant tuple
$(d, \ldots, d)$ by $d^k$.

\begin{definition}
  Let $f$ be a Boolean partial operation. We say (1) that $f$ is
  {\em self-dual} if $\overbar{x} \in \domain(f)$ for every $x \in
  \domain(f)$ and $f(x) = 1 - f(\overbar{x})$, where $\overbar{x}$
  denotes the complement of the tuple $x$, and (2) that $f$ is {\em
    idempotent} if $d^k \in \domain(f)$ and $f(d^k) = d$ for every $d \in D$.
\end{definition}

In the sequel, we will call a Boolean partial operation which is both
self-dual and idempotent a {\em pSDI-operation}, short for partial,
self-dual, and idempotent operation.
Let a {\em polymorphism pattern} of arity $r$ be a set of pairs
$(t,x)$ where $t$ is an $r$-ary tuple of variables and where $x$
occurs in $t$. We say that a $r$-ary partial operation $f$ over a set of
values $D$ {\em satisfies} an $r$-ary polymorphism pattern $P$ if
\[\domain(f) = \{(\tau(x_1),
\ldots, \tau(x_r)) \mid ((x_1, \ldots, x_r), x) \in P, \tau : \{x_1,
\ldots, x_r\} \rightarrow D\}\] and 
$f(\tau(x_1), \ldots, \tau(x_r)) = \tau(x)$ for every
$((x_1, \ldots, x_r), x) \in P$ and every $\tau : \{x_1, \ldots, x_r\}
\rightarrow D$.

A Boolean operation is pSDI if and only if it satisfies a
polymorphism pattern. To see this, note that if $f$ is pSDI, then it
is easy to create a polymorphism pattern $P$ by letting each tuple $t
\in \domain(f)$ such that $f(t) = 0$ correspond to an entry in
$P$. Similarly, it is not difficult to show that any partial operation
satisfying a polymorphism pattern must be self-dual and idempotent. We
then have the following link between sign-symmetric constraint
languages and partial operations satisfying polymorphism patterns.

\begin{theorem}\cite{lagerkvist2015}
  Let $f$ be a pSDI-operation. Then $\inv(f)$ is sign-symmetric. 
\end{theorem}

Hence, pSDI-operations provide a straightforward way to describe broad
classes of sign-symmetric constraint languages. It is also known that
if $\Gamma$ is sign-symmetric and $\SAT(\Gamma)$ is NP-hard, then
every partial polymorphism of $\Gamma$ is a subfunction of a
pSDI-operation preserving $\Gamma$~\cite{lagerkvist2015}[Theorem 3]
(see Lagerkvist~\cite{lagerkvist2016} for a full proof). We will now define
the pSDI-operations that will play a central role in our current
pursuit.

\begin{definition}
  Let $k \geq 2$. A $(k+1)$-ary partial operation is a {\em partial
    $k$-edge operation} if it satisfies the pattern consisting of 
  $((x, x, y, y, y, \ldots, y, y), y)$, $((x, y, x, y, y, \ldots, y,
  y),y)$, and for each $i \in \{4, \ldots, k+1\}$, the tuple $((y,\ldots,y, x, y, \ldots,
  y),y)$, where $x$ appears in position $i$.
\end{definition}

We will typically denote partial $k$-edge operations by $\edge_k$,
and, if the underlying set $D$ is important, by $\edge^D_k$. A partial
2-edge operation will sometimes be called a {\em partial Maltsev operation}.

\begin{definition}
  Let $k \geq 3$. A $k$-ary partial operation is a {\em partial
    $k$-ary near-unanimity operation} (partial $k$-NU operation) if it satisfies the pattern
  which for each $i \in
\{1, \ldots, k\}$ contains $((x, x, \ldots, x, y, x, \ldots,
x), x)$, where $y$ occurs in position $i$.
\end{definition}

We write $\near^{D}_k$ to denote this operation over the domain $D$,
and $\near_k$ if the domain is clear from the context, or not
relevant. Ternary partial NU-operations will sometimes be called {\em
partial majority operations}. Note that the partial majority operation
is total in the Boolean domain but is properly partial for every
larger domain. Last, we define the following class of self-dual
partial operations. Say that the argument $i$ of a $k$-ary partial
operation $f$ is {\em redundant} if there exists $j \neq i$ such that
$t[i] = t[j]$ for every $t \in \domain(f)$.


\begin{definition}
  Let $k \geq 2$. The {\em $k$-universal} operation $\universal_k$ is the
  Boolean $(2^{k} - 1)$-ary pSDI-operation defined on $2k + 2$ tuples such that (1)
  $\universal_k$ is not a partial projection and (2) $\universal_k$
  does not have any redundant arguments.
\end{definition}

While not immediate from the definition, the operation $\universal_k$
is in fact unique up to permutation of arguments. To see this, simply
take the $k$ non-constant tuples $t_1, \ldots, t_k \in
\domain(\universal_k)$ such that $\universal_k(t_1) = \ldots =
\universal_k(t_k) = 0$. Since $\universal_k$ is not a projection and
is pSDI, it follows that there cannot exist $i \in [2^{k} - 1]$ such
that $(t_1[i], \ldots, t_k[i]) = 0^k$. Hence, since $u_k$
does not have any redundant arguments, there for every
$t \in \{0,1\}^k \setminus \{0^k\}$ must exist a unique $i
\in [2^{k} - 1]$ such that $(t_1[i], \ldots, t_k[i]) = t$.

Last, we remark that there is a connection between our notion of
polymorphism patterns and the operations studied in connection to the
CSP dichotomy (see e.g.\ the survey by Barto et
al.~\cite{barto2017}). In technical terms polymorphism patterns
essentially matches {\em strong Maltsev condititions} where the right-hand side is
restricted to a single variable. Similar restrictions, called {\em
height-1 identities}, have been considered earlier and it is known
that the complexity of a CSP$(\Gamma)$ problem only depends on the
height-1 identities satisfied by the operations in
$\pol(\Gamma)$~\cite{barto2017b}.
\section{Structure of Constraint Languages under Minimal Restrictions}
\label{section:structure}

We now properly begin the first part of the paper, investigating the
structure of maximally expressive, yet restricted sign-symmetric
constraint languages. This investigation is performed via the study of
the weakest non-trivial pSDI-operations, including the operations
defined in Section~\ref{section:patterns}.  As a preview of the
structure, and of some of the included problems, we refer to
Figure~\ref{fig:1}.  The problem and language inclusions illustrated
in this figure will be shown across the next two subsections.

\begin{figure} 
\tikzset{
   problem/.style={rectangle,rounded corners,dotted, draw},
   class/.style={rectangle,rounded corners, draw}
}

\begin{tikzpicture}[shorten >=1pt, auto, node distance=1cm, ultra thick,
   edge_style/.style={draw=black, ultra thick}]



   \node [problem,align=left] (equations) {Subset Sum \\ Linear Equations \\ 1-in-$k$ SAT};

   \node [class, right = of equations] (2edge) {2-edge = 2-universal};

   \draw [->] (equations) edge (2edge);

   \node [class, below = of 2edge] (3edge) {3-edge} edge [<-] (2edge);

   \node [class, left = of 3edge] (3nu) {3-NU} edge [->] (3edge);

   \node [problem, left = of 3nu, align=left] (3nuapps) {2-SAT \\ $(d,2)$-CSP \\ Graph $k$-Clique} edge [->] (3nu);

   \node [class, right = of 3edge] (3uni) {3-universal} edge [<-] (3edge);
   
   \node [problem, right = of 2edge,align=left] (deg2) {Sidon Sets \\ Degree-2 Polynomials};
   \draw [->] (deg2) edge (3uni);

   \node [class, below = of 3nu] (4nu) {4-NU} edge [<-] (3nu);

   \node [problem, left = of 4nu, align=left] (3sat) {3-SAT \\ $(d,3)$-CSP \\ $(3,\ell)$-Hyperclique} edge [->] (4nu);

   \node [class, below = of 3edge] (4edge) {4-edge} edge [<-] (4nu);

   \draw [->] (3edge) edge (4edge);

   \node [class, below = of 3uni] (4uni) {4-universal} edge [<-] (3uni);

   \node [problem, right = of 3uni] (deg3) {Degree-3 Polynomials};
   \draw [->] (deg3) edge (4uni);

   \draw [->] (4edge) edge (4uni);

   \node [below = of 4nu] (dotnu) {\dots} edge [<-] (4nu);

   \node [below = of 4edge] (dotedge) {\dots} edge [<-] (4edge);
 
   \draw [->] (dotnu) edge (dotedge);

   \node [below = of 4uni] (dotuni) {\dots} edge [<-] (4uni);

   \draw [->] (dotedge) edge (dotuni);

   \node [class, below = of dotnu] (knu) {$k$-NU} edge [<-] (dotnu);

   \node [class, below = of dotedge] (kedge) {$k$-edge} edge [<-] (dotedge);

   \draw [->] (knu) edge (kedge);

   \node [problem, below = of 3sat, align=left] {$(k-1)$-SAT \\ $(d, k-1)$-CSP \\ $(k-1,\ell)$-Hyperclique} edge [->] (knu);

   \node [class, below = of dotuni] (kuni) {$k$-universal} edge [<-] (dotuni);
  
   \draw [->] (kedge) edge (kuni);

   \node [problem, right = of dotuni] (degk) {Degree-$(k-1)$ Polynomials};
   \draw [->] (degk) edge (kuni);

\end{tikzpicture}
\caption{The inclusion structure between
  selected minimal pSDI-operations (solid outlines),
  and some problems that reduce
  to the corresponding SAT or CSP problem (dotted outlines).
  Several classes on each level $k \geq 3$ have been omitted.}
\label{fig:1}
\end{figure}
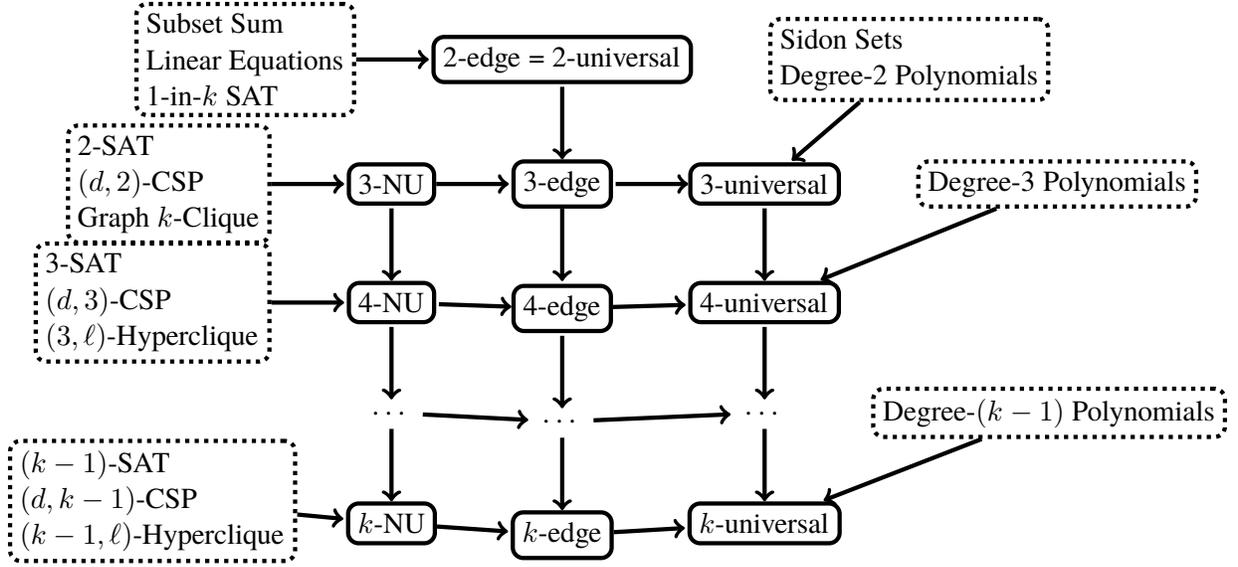

More precisely, by ``weakest'' pSDI-operations, we mean 
partial operations that are minimal in the following sense. 
Recall that for every pSDI-operation $f$ and
every subfunction $f'$ of $f$, we 
have $\inv(f) \subseteq \inv(f')$. 
This motivates the following definition. 

\begin{definition}
  \label{def:psdi-minimal}
  Let $f$ be a pSDI-operation. We say that $f$ is \emph{trivial}
  if it is a subfunction of a projection, and a
  \emph{minimal non-trivial} pSDI-operation if $f$ is non-trivial
  but every proper subfunction $f'$ of $f$ which is a pSDI-operation
  is trivial.
\end{definition}

Our study in this section is focused on constraint languages
$\Gamma=\inv(f)$ where $f$ is a single minimal non-trivial
pSDI-operation, since these are the most expressive sign-symmetric
constraint languages that are still restricted in expressive power. 
We begin by giving some examples for the particular classes of $k$-NU,
$k$-edge and $k$-universal partial operations defined in 
Section~\ref{section:patterns}. 

\subsection{Properties of specific sign-symmetric constraint  languages}


In this section, we provide some illustrative examples of languages
included in $\inv(f)$ for particular pSDI-operations $f$. We first
recall the following result from Lagerkvist \& Wahlstr\"om.


\begin{theorem}~\cite{lagerkvist2017c} \label{thm:linear}
  Let $F$ be a finite set of partial operations such that
  $\inv(F)$-SAT is NP-complete. Then any $n$-ary Boolean relation
  has a pp-definition over $\inv(F)$ using at most $O(n)$
  existentially quantified variables.
\end{theorem}

In effect, this implies that any constraint language $\inv(F)$, where
$F$ is a finite set of pSDI-operations, is extremely expressive. 
One direct consequence is that $\inv(F)$ contains at least
$2^{2^{cn}}$ $n$-ary relations for some constant $0 < c \leq 1$.
This makes such constraint languages markedly different from finite
constraint languages, since for any finite constraint language $\Gamma$,
the number of $n$-ary qfpp-definable relations over $\Gamma$ is bounded by
$O(2^{p(n)})$ for a polynomial $p$ depending on $\Gamma$. 
This also implies that there cannot exist a
finite $\Gamma$ such that $\ppol(\Gamma) =\strongof{F}$.
In fact, the relations of $\inv(F)$ for such an $F$ are dense enough
that for any $n$-ary relation $R$, a random padding of $R$ by $O(n)$
parity-check variables is enough to create a variable in $\inv(F)$
with high probability. This fact will be exploited in
Section~\ref{section:padding_constants}.

Despite this, we will see that the pSDI-operations defined in
Section~\ref{section:patterns} do correspond roughly to natural
restrictions on the expressive power of a language $\Gamma$.  We now
illustrate the classes with a few examples. In the process will
occasionally refer to the language inclusions illustrated in
Figure~\ref{fig:1}. Proofs of these inclusions is given in
Theorem~\ref{thm:inclusions} in Section~\ref{section:inclusions}.  Let
us now begin with a basic example.

\begin{lemma}   \label{lemma:near_inv}
  $R \in \inv(\near_k)$ for every $(k-1)$-ary relation $R$, $k \geq 3$.
\end{lemma}
\begin{proof}
  Let $t_1, \ldots, t_k \in R$ be such that $\near_k(t_1,\ldots,t_k)$
  is defined, and for $i \in [k-1]$ let $t^{(i)}=(t_1[i], \ldots, t_k[i])$.
  For every $i \in [k-1]$, either $t^{(i)}$ is constant or there is a 
  single index $j$ where $t^{(i)}[j]$ deviates from its other
  entries. By the pigeonhole principle, there is at least one 
  index $j \in [k]$ such that $t^{(i)}[j]$ does not deviate from the 
  majority for any $i \in [k-1]$. Then we have 
  $\near_k(t_1,\ldots,t_k)=t_j$.
\end{proof}

We also show a corresponding negative statement. By the inclusions
shown in the next section, this will imply that a $k$-clause is not
preserved by any operation at ``level $k$'' of the hierarchy in Figure~\ref{fig:1}. 

\begin{lemma} \label{lemma:kary-knu}
  Let $R \subset \{0,1\}^k$ be a $k$-clause, i.e., 
  $|R|=2^k-1$, $k \geq 2$. 
  Then $R$ is not preserved by the partial $k$-universal operation. 
\end{lemma}
\begin{proof}
  By sign-symmetry, we assume that $R=\{0,1\}^k \setminus \{0^k\}$. 
  Let $t_1, \ldots, t_k$ be the non-constant tuples in $\domain(\universal_k)$
  such that $\universal_k(t_i)=0$ for each $i \in [k]$. 
  Then for each $i \in [2^k-1]$, the tuple $t^{(i)}=(t_1[i],\ldots,t_k[i])$
  defines a tuple of $R$; thus the application
  \[
  \universal_k(t^{(1)}, \ldots, t^{(2^k-1)})=0^k
  \]
  is defined and shows that $R \notin \inv(\universal_k)$. 
\end{proof}

Next, we consider a canonical example of a useful relation preserved
by the partial $2$-edge operation. 

\begin{lemma} \label{lemma:2edge:linear}
  Let $R(x_1,\ldots,x_n) \subseteq \{0,1\}^n$ be defined via a linear
  equation 
  \[
  \sum_{i=1}^n \alpha_i x_i = \beta
  \]
  evaluated over a finite field $\mathbb{F}$. Then $R \in \inv(\edge_2)$. 
\end{lemma}
\begin{proof}
  This is a special case of the notion of a \emph{Maltsev embedding} of $R$
  previously investigated by the authors~\cite{LagerkvistW17CP}. It is
  known that a relation with a Maltsev embedding is closed under a
  family of partial operations, of which $\edge_2$ is the simplest.
\end{proof}

A particular example of such relations is the \textsc{Exact SAT} problem.
We show that its 1-in-$k$ relations are also not closed under $\near_k$.

\begin{lemma} \label{lemma:exactsat-2edge}
  Let $R_{1/k} = \{(x_1, \ldots, x_k) \in \{0,1\}^k \mid x_1 + \ldots
  + x_k = 1\}$, and $\Gamma_{\mathrm XSAT} = \{R^{s}_{1/k} \mid k \geq
  1, s$ is a $k$-ary sign-pattern$\}$. Then $\Gamma_{\mathrm XSAT}
  \subseteq \inv(\edge_2)$ but is not preserved by $\near_k$ for any $k$.
\end{lemma}
\begin{proof}
  The positive direction follows from Lemma~\ref{lemma:2edge:linear},
  since $R_{1/k}$ can be phrased as a linear equation over the
  integers mod $p$, for $p \geq k+1$. The negative direction is
  immediate: let $R_{1/k}=\{t_1,\ldots,t_k\}$. Then
  $\near_k(t_1,\ldots,t_k)$ is defined and equals $0^k$. 
\end{proof}

Another example of a problem with the character of linear equations is
\textsc{Subset Sum}. Even though an instance of \textsc{Subset Sum} is
defined by just a single linear equation rather than as a
$\SAT(\Gamma)$ instance, we show in Section~\ref{section:subsetsumhard}
that the complexity of \textsc{2-edge-SAT} and \textsc{Subset Sum}
are closely connected. 
As for the class $\inv(\edge_k)$ for $k \geq 3$, 
the inclusions illustrated in Figure~\ref{fig:1} 
imply that this class contains both
relations with linear equation extensions and 
all $(k-1)$-clauses. 

Finally, we show two examples for the partial $k$-universal operation $\universal_k$. 
The first is a previously studied class of Lokshtanov et
al.~\cite{LokshtanovPTWY17SODA}. Note this problem does admit an
improved algorithm. 

\begin{definition}
  Let $P_d$ denote the set of Boolean relations such that each $n$-ary
  $R \in P_d$ is the set of roots of an $n$-variate polynomial
  equation where each polynomial has degree at most $d$.
\end{definition}

\begin{lemma} \label{lemma:dpoly-d+1univ}
  Let $R \in P_d$ be an $n$-ary relation. Then $R$ is
  preserved by $\universal_{d+1}$, but not by any other
  non-trivial pSDI-operation of domain size at most $2d+2$.
\end{lemma}
\begin{proof}
  For the first direction, let $P(x_1,\ldots,x_n)$ be the polynomial defining $R$, and  
  let $t_1, \ldots, t_r \in R$ be such that $\universal_{d+1}(t_1,\ldots,t_r)=t'$ is defined.
  Since the set of relations representable by bounded-degree polynomials 
  is sign-symmetric, we may assume for simplicity that $t'=1^n$. 
  The tuples $(t_1, \ldots, t_r)$ define a new polynomial 
  of degree at most $d$ and with at most $d+1$ variables, defined by 
  identifying all pairs of variables $x_i$ and $x_j$ that have the same pattern
  in $(t_1, \ldots, t_r)$, i.e., if $t_a[i]=t_a[j]$ for every $a \in [r]$. 
  We also eliminate any variable $x_i$ such that $t_j[i]=1$ for every $j \in [r]$ 
  by replacing $x_i$ by  the constant $1$ in $P$. 
  Let $P'$ be the resulting polynomial, and let $R'$ be the corresponding relation. 
  If $\ar(R')<d+1$, then by Lemma~\ref{lemma:kary-knu} $R'$ is
  preserved by $\near_{d+1}$ and thus by $\universal_{d+1}$ as well
  (see Theorem~\ref{thm:inclusions}). 
  Otherwise, for each $I \subset [d+1]$ let $\alpha_I$ be the coefficient
  of the monomial $\prod_{i \in I} x_i$ in $P'$, 
  and let $\chi_I \in \{0,1\}^{d+1}$ be the tuple such that $\chi_I[i]=1$ 
  if and only if $i \in I$. Note that $P'(\chi_I)= \sum_{I' \subseteq I} \alpha_{I'}$.
  We find that $\alpha_I=0$ for every $I$. Indeed, $\alpha_{\emptyset}=0$ 
  since $0^{d+1} \in R'$; and $\alpha_{\{i\}}=0$ for every $i$ since 
  $P'(\chi_{\{i\}})=\alpha_{\{i\}}+\alpha_{\emptyset}=\alpha-{\{i\}}=0$;
  and so on, in order of increasing cardinality of $I$. 
  Then $P'$ is the constantly-zero polynomial, and $1^{d+1} \in R'$,
  hence $t' = 1^n \in R$. We have thus shown that relations defined as
  roots of polynomials of degree $d$ are preserved by the $(d+1)$-universal
  operation.

  In the other direction, the same argument will show that for any 
  pSDI-operation $f$ with $|\domain(f)| \leq 2d+2$ other than
  the $(d+1)$-universal operation, it is possible to define 
  a polynomial on $(|\domain(f)|-2)/2$ variables and of degree
  at most $d$ such that the corresponding relation is not preserved
  by $f$. Indeed, let $n=(|\domain(f)-2|)/2$ and $r=\ar(f)$, 
  and let $t_1, \ldots, t_r$ be tuples of arity $n$ such that 
  no tuple $(t_1[i], \ldots, t_r[i])$ is constant and 
  $f(t_1,\ldots,t_r)=1^n$ is defined. 
  If $n \leq d$, then we may simply consider the 
  polynomial $P(x_1,\ldots,x_n)=\prod_{i \in [n]} x_i$,
  whose corresponding relation $R$ is not preserved by $f$.
  Otherwise, let $I \subset [d+1]$ be such that
  $\chi_I \notin \{t_1,\ldots,t_r\}$; this exists
  since $f$ is not the $(d+1)$-universal partial operation.
  Let $P'$ be the $d+1$-variate polynomial
  with coefficients $\alpha_{J}=0$ if $I \not \subseteq J$,
  and with $\alpha_J=(-1)^{|J|-|I|}$ otherwise, for all $J \subset [d+1]$.
  Then $P'(t_I)=1$, and it can be verified that $P'(t_J)=0$
  for every $J \subset [d+1]$, $J \neq I$, whereas 
  $P'(1^{d+1})=-(-1)^{d+1-|I|}$. Hence the relation
  corresponding to $P'$ is not preserved by $f$.
\end{proof}

Finally, we give one example of a symmetric relation in
$\inv(\universal_3)$ that has no obvious connection to roots of
polynomials. A \emph{Sidon set} is a set $S \subseteq \{0,\ldots,n\}$
in which all sums $i+j$, $i, j \in S$ are distinct. 

\begin{lemma} \label{lemma:sidon-uni3}
  Let $S \subseteq \{0,\ldots,n\}$ be a Sidon set,
  and define a relation $R(x_1, \ldots, x_n) \subseteq \{0,1\}^n$ as
  \[
  R(x_1,\ldots, x_n) \equiv (\sum_{i=1}^n x_i \in S).
  \]
  Then $R$ is preserved by $\universal_3$. 
\end{lemma}
\begin{proof}
  Assume that there exists $t_1, \ldots, t_7 \in R$ such that
  $\universal_3(t_1, \ldots, t_7) = t \notin R$. For $i \in [n]$, 
  let $x_i=(t_1[i], \ldots, t_7[i])$ be the tuple of values
  taken by argument $i$ of $R$ in these tuples. 
  Then the tuples $x_i$ take up to 8 different values,
  partitioned as two constant tuples and three pairs of 
  complementary tuples. Let $X_j$ for $j=1, 2, 3$ be the 
  set of arguments $i \in [n]$ such that the tuple $x_i$ 
  belongs to the $j$:th of these pairs, and let 
  $n_j$ be the difference in Hamming weight compared to $t$ 
  if flipping all values belonging to $X_j$. 
  Let $W$ be the Hamming weight of $t$. 
  Then $S$ contains the values $W+n_1$, $W+n_2$, $W+n_1+n_3$
  and $W+n_2+n+3$, forming two pairs of weights with
  common difference $n_3$. Since $n_3 \neq 0$, 
  we must have $n_1=n_2$. By symmetry, we have $n_1=n_2=n_3$. 
  But then $S$ contains the values $W+n_1$, $W+n_1+n_2=W+2n_1$, 
  and $W+n_1+n_2+n_3=W+3n_1$, which is a contradiction. 
  Thus $n_j=0$ for at least one $j$, hence $W \in S$ 
  and $t \in R$, contradicting the original assumption.
\end{proof}

%

\subsection{Structure of minimal non-trivial pSDI-operations}
\label{section:inclusions}

Note that if $f$ is a pSDI-operation, then $|\domain(f)|=2k+2$
for some $k$, since $f$ is defined on the two constant tuples
and since the tuples of the domain can be paired up
as $(t, \overbar{t})$ where $\overbar{t}$ is the complement of $t$.
Hence, we define the \emph{level} of a minimal 
non-trivial pSDI operation $f$ as $(|\domain(f)|-2)/2$. 
We find no examples on level 0 or 1, and  
the only non-trivial example on level 2 is
the 2-edge operation. At each level $k \geq 3$ 
the partial $k$-NU and $k$-universal operations
are the unique strongest and weakest minimal non-trivial
pSDI-operation, respectively, whereas the $k$-edge
operation is intermediate.
This structure is also illustrated in Figure~\ref{fig:1}. 
We also find that the $k$-universal operations $\universal_k$ 
are maximally weak in the sense that any non-trivial pSDI-operation 
with a domain of size $2k + 2$ can define $\universal_k$. 

We begin with the following lemma, which formalizes one of the main 
methods of constructing a $(k+1)$-ary partial operation from a $k$-ary 
partial operation. We refer to $g$ as an \emph{argument padding} of $f$. 

\begin{lemma} \label{lemma:padding}
  Let $f$ be a $k$-ary partial operation and let $g$ be a $(k+1)$-ary
  partial operation such that (1) $\pro_{1, \ldots, k}(\domain(g)) =
  \domain(f)$ and (2) $f(x_1, \ldots, x_k) = g(x_1, \ldots, x_k, x_{k+1})$
  for every $(x_1, \ldots, x_k, x_{k+1}) \in \domain(g)$. Then $g \in \strongof{f}$.
\end{lemma}

\begin{proof}
  Let $f$ and $g$ be as in the statement, and first construct the
  $(k+1)$-ary partial operation \[f'(x_1, \ldots, x_k, x_{k+1}) =
  f(\pi^{k+1}_1(x_1, \ldots, x_k, x_{k+1}), \ldots, \pi^{k+1}_{k}(x_1,
  \ldots, x_{k}, x_{k+1})).\] Clearly, $f' \in \strongof{f}$, since it
  is a composition of $f$ and the projections $\pi^{k+1}_1, \ldots,
  \pi^{k+1}_k$, and it is not difficult to see that $\pro_{1, \ldots, k}(\domain(f')) =
  \domain(f)$ and that $g$ can be obtained as a subfunction of
  $f'$. Since $\strongof{f}$ is closed under taking
  subfunctions it follows that $g \in \strongof{f}$.
\end{proof}

The following will aid us in reasoning about
minimal non-trivial pSDI-operations. 

\begin{lemma} \label{lemma:whichminimal}
  Let $f$ be a pSDI-operation with $|\domain(f)|=2k+2$, $k \geq 3$. 
  Then $f$ is a minimal non-trivial operation if and only if
  $f$ is an argument padding of $\near_k$. 
\end{lemma}
\begin{proof}
  In the one direction, assume that $f$ is a padding of $\near_k$.
  It is not hard to verify that every subfunction $f'$ of $f$ which is pSDI
  is a partial projection, and that $f$ is non-trivial. Thus,
  $f$ is minimal non-trivial. In the other direction, 
  assume that $f$ is minimal and non-trivial, and let $r=\ar(f)$.
  Let $t_1, \ldots, t_k$ be the non-constant tuples 
  such that $f(t_1)=0$ is defined. For each $i \in [k]$,
  let $j_i \in [r]$ be such that making $f$ 
  undefined on $t_i$ and its complement $\overbar{t_i}$
  leaves a subfunction of $\pi_{j_i}^r$. It follows 
  that for all $a \in [k]$, $t_a[j_i] \neq 0$
  if and only if $a=i$. Then the arguments $j_1$,
  \ldots, $j_k$ of $f$ define the partial $k$-NU
  operation, and $f$ is a padding of it. 
\end{proof}

Our claims about the weakest and strongest operations
follows from this.

\begin{lemma} \label{lemma:weakeststrongestminimal}
  The following hold.
  \begin{enumerate}
  \item The unique non-trivial non-total pSDI-operation at level $k<3$
    is the partial 2-edge operation.
  \item For any minimal non-trivial pSDI-operation $f$
    at level $k \geq 3$, we have
    $
    \inv(\near_k) \subseteq \inv(f) \subseteq \inv(\universal_k).
    $
  \item There are at most $2^{2^k-k-1}$ distinct
    minimal non-trivial pSDI-operations at level $k$.  
  \end{enumerate}
\end{lemma}
\begin{proof}
  \emph{1.} It is easy to verify that no non-trivial operation is
  possible on level 1. Let $f$ be a non-trivial pSDI-operation
  on level 2, and let $t_1, t_2 \in \domain(f)$ be the non-constant tuples
  such that $f(t_1)=0$. Consider the options for the pairs 
  $(t_1[i], t_2[i])$ for $i \in [\ar(f)]$. 
  If two distinct positions $i, i'$ give identical pairs, then
  $t[i]=t[i']$ for every $t \in \domain(f)$ and $i$ and $i'$
  are redundant arguments in $f$, which we may assume does not occur. 
  If $t_1[i]=t_2[i]=0$ for some $i \in [\ar(f)]$ then $f$ is a partial projection. 
  This leaves three possible arguments, and unless all three exist,
  $f$ will be a total operation. The remaining case is that
  $f=\edge_2$. 

  \emph{2.} By Lemma~\ref{lemma:whichminimal} $f$ is a padding of
  $\near_k$, which provides the first inclusion.
  For the second, we may assume that $f$ has no redundant
  arguments, since otherwise $f$ is equivalent to
  an operation with fewer arguments. But then by design,
  $\universal_k$ is a padding of $f$, and the second inclusion
  follows. 

  \emph{3.} By Lemma~\ref{lemma:whichminimal}, we can restrict
  our attention to paddings of $\near_k$. Since $f$ is a
  pSDI-operation, it is defined by the values of the $k$ non-constant
  tuples $t$ in the domain with $f(t)=0$. 
  Let $t_1, \ldots, t_k$ be those tuples, 
  and for $i \in [\ar(f)]$ let $t^{(i)}=(t_1[i], \ldots, t_k[i])$. 
  As above,  we may assume that $t^{(i)} \neq t^{(j)}$ for
  all distinct $i, j \in [\ar(f)]$. This leaves at most $2^k$ possible 
  arguments. Furthermore, $t^{(i)}$ cannot be all-zero unless $f$ is
  a partial projection, and $k$ arguments are determined by $\near_k$. 
  This leaves $2^k-k-1$ arguments, whose presence or absence defines $f$. 
\end{proof}

The inclusion structure between the $k$-NU, 
$k$-edge and $k$-universal partial operations are now straightforward
to prove with these results.

\begin{theorem} \label{thm:inclusions}
  Let $k \geq 3$. Then the following inclusions hold.
  \begin{enumerate}
    \item
      $\inv(\edge_2) \subset \inv(\edge_k)$,
    \item
      $\inv(\near_k) \subset \inv(\edge_k) \subset
      \inv(\universal_k)$,
    \item
      $\inv(\near_k) \subset \inv(\near_{k+1})$,
    \item
      $\inv(\edge_k) \subset \inv(\edge_{k+1})$, and
    \item
      $\inv(\universal_k) \subset \inv(\universal_{k+1})$.
  \end{enumerate}
\end{theorem}

\begin{proof}
  For the inclusions, the second item follows from Lemma~\ref{lemma:weakeststrongestminimal},
  and every other inclusion follows from Lemma~\ref{lemma:padding}. 
  Indeed, it is readily verified that for every $k \geq 3$, $\edge_k$ 
  is an argument padding of $\edge_{k-1}$ and $\near_{k+1}$ is 
  an argument padding of $\near_k$. For the universal operations,
  let $t_1, \ldots, t_{k+1}$ be the non-constant tuples of 
  $\domain(\universal_{k+1})$ such that $\universal_{k+1}(t_i)=0$, $i \in [k+1]$. 
  Then the tuples $t^{(i)}=(t_1[i], \ldots, t_{k+1}[i])$, $i \in [2^{k+1}-1]$
  spell out all $(k+1)$-tuples except $0^{k+1}$, without repetition. 
  Consider the subset $I \subset [\ar(\universal_{k+1})]$ consisting
  of indices $i$ such that $t_{k+1}[i]=0$. Note that $t^{(i)}$ for $i \in I$
  enumerates all $k$-tuples except $0^k$, padded with a $0$. It follows that
  $\pro_I(\universal_{k+1})=\domain(\universal_k)$ and that $\universal_{k+1}$
  is an argument padding of $\universal_k$. By Lemma~\ref{lemma:padding}
  the inclusion follows. 
  
  To show that the inclusions are strict, consider the following:
  a $k$-clause is preserved by $\near_{k+1}$ (Lemma~\ref{lemma:near_inv}) 
  but not by $\universal_k$ (Lemma~\ref{lemma:kary-knu});
  a 1-in-$k$ constraint is preserved by $\edge_2$ but not by $\near_k$
  (Lemma~\ref{lemma:exactsat-2edge}); 
  and the language $P_{k-1}$  of roots of polynomials of degree at
  most $k-1$ is preserved by~$\universal_{k}$ but not by any other
  operation on level~$k$ by Lemma~\ref{lemma:dpoly-d+1univ}.
\end{proof}

Finally, we have an easy consequence in more general terms.

\begin{corollary}
  Let $f$ be a pSDI-operation with $|\domain(f)|=2k+2$. 
  Then $\inv(f) \subseteq \inv(\universal_k)$. 
\end{corollary}
\begin{proof}
  Let $f'$ be an arbitrary minimal pSDI-operation
  that is a subfunction of $f$. Then $f'$ belongs to
  some level $k' \leq k$, hence 
  $\inv(f) \subseteq \inv(\universal_{k'}) \subseteq   \inv(\universal_k)$
  by Lemma~\ref{lemma:weakeststrongestminimal} and Theorem~\ref{thm:inclusions}. 
\end{proof}

\subsection{Complementary consequences}
\label{section;comp-cons}

We now consider some dual questions, i.e., what consequences can we
(in general) draw from the information that some sign-symmetric
language $\Gamma$ is not preserved by $f$, for some pSDI-operation
$f$? We begin with an easy result, which forms the building block of 
later results. 

\begin{lemma} \label{lemma:consequence-not-knu}
  Let $\Gamma$ be a sign-symmetric language which is not preserved by
  $\near_k$, for some $k \geq 3$. Then $\Gamma$ can qfpp-define a
  $k$-ary symmetric relation $R$ such that $R$ does not contain tuples
  of weight 0, but does contain tuples of weight $2$. 
\end{lemma}
\begin{proof}
  Let $k \geq 3$ be an arbitrary constant, and let $R \in \Gamma$
  be a relation not preserved by $\near_k$ of some arity $n=\ar(R)$. 
  Let $t_1, \ldots, t_k \in R$ be witnesses to this, i.e., 
  $t=\near_k(t_1,\ldots,t_k)$ is defined and $t \notin R$. 
  Define $t^{(i)}=(t_1[i], \ldots, t_k[i])$.

  By sign-symmetry, we may assume that $t=0^n$.
  Furthermore, if there is an argument $i \in [n]$ such that 
  $t^{(i)}=0^k$, then we can find a smaller
  counterexample by fixing argument $i$ of $R$ to be constantly 0. 
  Thus, for every $i \in [n]$, the tuple $t^{(i)}$ now contains
  precisely one non-zero value. 
  Let us define a new relation $R'(x_1, \ldots, x_k)$ of arity $k$ 
  by identifying arguments according to this, i.e.,
  for every position $i \in [n]$ such that $t^{(i)}$ is
  non-zero in position $j \in [k]$, insert variable $x_j$
  in position $i$ in $R$. Additionally define $R''$
  as the result of the conjunction of all $k!$ applications of $R'$ 
  with permuted argument order. Then $R''$ is a symmetric relation
  which contains all tuples of weight $1$ but none of weight 0. 
  Thus, $\Gamma$ qfpp-defines a relation $R_k=R''$ 
  as described of every arity $k \geq 3$. 
\end{proof}

By a similar strategy, we have an important result about languages not
preserved by the $k$-universal operation. 

\begin{lemma} \label{lemma:notkuniv-kcnf}
  Let $\Gamma$ be a sign-symmetric language not preserved by
  $\universal_k$ for some $k \geq 2$. Then $\Gamma$ can qfpp-define 
  all $k$-clauses. 
\end{lemma}
\begin{proof}
  Let $R \in \Gamma$ be a relation not preserved by $\universal_k$,
  and let $n=\ar(R)$ and $r=2^k-1$ be the arity of $\universal_k$. 
  Let $t_1, \ldots, t_r \in R$ be such that
  $\universal_k(t_1, \ldots, t_r) = t$ is defined and $t \notin R$. 
  By sign-symmetry of $\Gamma$, we may assume $t=0^n$. 
  Create a new relation by identifying all variables $x_i$ and $x_j$
  in $R(x_1,\ldots,x_n)$ for which $t_a[i]=t_a[j]$ for every $a \in [r]$.
  Also assume that there is no variable $x_i$ such that $t_a[i]=0$
  for every $a \in [r]$, or else replace $x_i$ by the constant $0$ in
  $R$ (again by sign-symmetry). This defines a new relation $R'$
  of arity at most $k$.  Since $t \notin \{t_1,\ldots,t_r\}$, we find 
  that $R'$ has arity precisely $k$
  and contains every possible $k$-tuple except $0^k$, i.e., 
  $R'$ qfpp-defines a $k$-clause. By sign-symmetry, $\Gamma$ 
  qfpp-defines all $k$-clauses. 
\end{proof}

\subsubsection{Infinitary case}

Finally, we consider consequences of a language not being preserved by  
any operation in a family of operations. 

\begin{theorem} \label{thm:not-any-knu}
  Let $\Gamma$ be a sign-symmetric language that is not 
  preserved by the partial $k$-NU operation, for any $k$.
  Then one of the following holds.
  \begin{enumerate}
  \item $\Gamma$ can qfpp-define all $k$-clauses for every $k$.
  \item $\Gamma$ can qfpp-define 1-in-$k$-clauses for every $k$.
  \item There is a fixed prime $p$ such that $\Gamma$ can
    qfpp-define relations
    \[
    \sum_{i=1}^k x_i \equiv a \pmod p
    \]
    for every $0 \leq a < p$, of every arity $k$.
  \end{enumerate}
\end{theorem}

Before we proceed with the proof, let us make a simple observation
about qfpp-definitions among symmetric relations. 

\begin{lemma} \label{lemma:symmetric-tools}
  Let $R$ be a symmetric $n$-ary relation,
  including tuples of weights $S \subseteq \{0,\ldots,n\}$.
  Using $R$, we can qfpp-define symmetric relations
  of the following descriptions.
  \begin{enumerate}
  \item Shift down: a relation of arity $n-1$
    accepting values $S'=\{x-1 \mid x \in S, x>0\}$.
  \item Truncate: a relation of arity $n-1$
    accepting values $S'=\{x \in S \mid x < n\}$.
  \item Grouping: for any integer $p>1$, 
    a relation of arity $\lfloor n/p \rfloor$
    accepting values $S'=\{x' \mid x'p \in S\}$. 
  \end{enumerate}
\end{lemma}
\begin{proof}
  These are implemented by, respectively, fixing an argument to $1$ in $R$;  
  fixing an argument to $0$ in $R$; and grouping arguments of $R$ in
  groups of size $p$ (after truncating $\ar(R)$ to an even multiple of $p$).
\end{proof}

We can now show the result. 

\begin{proof}[Proof of Theorem~\ref{thm:not-any-knu}]
  Let $k \geq 3$ be an arbitrary constant. By Szemer\'edi's theorem~\cite{SzemerediThm}
  there is a constant $n=N(2k, 1/2(k+1))$ such that every set
  $S \subseteq [n]$ with $|S| \geq n/(2k+1)$ contains an arithmetic
  progression $a, a+p, \ldots$ of at least $2k$ items. 
  Let $R_n$ be a relation produced by Lemma~\ref{lemma:consequence-not-knu}
  of arity $n$, and let $S$ be the accepted weights for $R_n$.
  Say that an arithmetic progression $a, a+p, \ldots$ is
  \emph{complete} in $S$ if $S$ contains all values
  $\{x \in \{0,\ldots,n\} \mid x \equiv a \pmod b\}$.
  We consider a few cases. 

  \emph{Case: $S$ contains an incomplete arithmetic progression
    with at least $k$ items.} We show that in this case,
  $\Gamma$ can qfpp-define all $k$-clauses.
  Let $a, a+p, \ldots, a+(k-1)p \in S$ be an arithmetic progression 
  that in one direction does not continue.
  If $a \geq p$ and $a-p \notin S$, then by shifting,
  truncating and grouping we can qfpp-define
  the $k$-ary relation $(\sum_i x_i \geq 1)$;
  in the other case, if $a+kp \leq n$ and $a+kp \notin S$,
  then we can similarly qfpp-define the $k$-ary
  relation $(\sum_i x_i < k)$. In both cases, taking
  closure under sign-symmetry shows that we can
  qfpp-define all $k$-clauses.
  This finishes this case.
  
  \emph{Case: $S$ is sparse.} 
  Assume that $|S|<n/(2k+2)$ and that $S$ contains no incomplete 
  arithmetic progression of at least $k$ items. 
  By truncation, we can assume that $n$ is an even multiple of $k+1$.
  By self-intersecting $R_n$ by its shifted variant,
  if needed repeated up to $k$ times,
  we can further ensure that $S$ contains no pairs $i, i+1$,
  except possibly in a chain $n-i$, $n-i+1$, \ldots
  ending with $n$, while retaining that $n$ is a multiple of $k$.
  By only doing this as many times as needed, we can be sure 
  that there is at least one \emph{isolated} weight $w$, $0<w<n$, 
  such that tuples of weight $w$ are accepted but not $w-1$ or $w+1$
  (recall that we start with a relation with $1 \in S$).
  
  Now partition $\{0,\ldots,n\}$ into windows
  $(0, \ldots, k)$, $(k+1, \ldots, 2k+1)$, \ldots
  of length $k+1$. By the density of $S$ (which did not increase
  during our modifications), at least half the windows
  contain no elements.  We may safely assume $n \geq 5k$; 
  thus there is an empty window that is not the first or the last. 
  Let $w$ be an isolated weight. Then by sliding the window
  containing $w$ towards the internal empty window,
  we must eventually reach a window where there is an isolated
  weight which is either in position $1$ or $k-1$ of the window.
  This lets us qfpp-define either a 1-in-$k$-clause or a
  $(k-1)$-in-$k$-clause; and in the latter case we get a
  1-in-$k$-clause by negating all variables. 
  Thus if $S$ is sparse and contains no incomplete progressions of
  length $k$, we can qfpp-define a 1-in-$k$-clause.

  \emph{Case: $S$ is dense.}
  Finally, we assume that $|S| \geq n/(2k+2)$ but does not contain any
  incomplete progressions of length $k$. By Szemer\'edi's theorem,
  $S$ contains at least one complete progression
  $\{x \mid 0 \leq x \leq n, x \equiv a \pmod p\}$
  for some $a$ and $p$, with at least $2k$ entries 
  (i.e., $(2k-1)p \leq n$). 
  Let $R_n^{-i}$ be the $(n-i)$-ary relation produced by 
  shifting $R_n$ $i$ steps down and consider the relation
  $R_n' = R_n \land R_n^{-p} \land \ldots \land R_n^{-(k-1)p}$
  of arity $n-(k-1)p$,
  with applications of $R_n$ and $R_n^{-i}$ padded with zeroes as necessary.
  Then $R_n'$ is the union of complete progressions with 
  difference $p$, since every weight $w$ accepted by $R_n'$
  corresponds to a progression of length $k$ in $S$. 
  Furthermore, the same holds for any constant shift $R_n'^{-i}$
  of $R_n'$, $i<p$. Note that $R_n'$ still has arity at least $pk$. 

  Let $A \subseteq \{0, \ldots, p-1\}$ be the weights $a$
  such that $R_n'$ contains the complete progression with offset $a$.
  Note that $0 \notin A$. By shifting and self-intersecting
  we can reduce to the case that $|A|=1$, i.e.,
  the remaining relation is equivalent to $(\sum x_i \equiv a \pmod p)$
  for some $a$ and $p$. 

  If $p \geq k$, then clearly $R_n'$ qfpp-defines a 1-in-$k$ relation
  by further shifting and truncation. Thus, if the difference $p$ 
  of the relations produced this way can grow without bound,
  then $\Gamma$ qfpp-defines 1-in-$k$ relations of all arities $k$. 

  Otherwise, if none of the above cases applies infinitely often,
  then there is a fixed $p$ such that
  this process produces relations
  $(\sum x_i \equiv a \pmod p)$ of infinitely many
  arities $k$, which leads to the last case in the theorem. 
  Assume we are in this case. 
  If $p$ is not a prime, we fix a prime $p'$  that
  divides $p$, and let $a'=a \bmod (p/p')$. 
  Shift the relation down by $a'$ and group
  the variables into blocks of size $p/p'$. 
  Then the remaining relation is equivalent to
  $(\sum_i x_i \equiv a'' \pmod p')$ for some $a''$.
  By shifting, and by starting from a sufficiently 
  large relation with period $p'$, we can produce 
  all relations as in the last case in the theorem.
\end{proof}

Finally, we note that since $k$-clauses can qfpp-define the other two
kinds of clauses, the same statement holds with only 1-in-$k$ clauses
and the counting relations $\mathrm{mod}\, p$.

\paragraph{Section summary.}
In summary of this section, towards the purpose of discussing
sign-symmetric languages $\Gamma$ such that $\SAT(\Gamma)$ does, or
does not, admit an improved algorithm under SETH, we conclude the
following. Recall that $\sat{k}$ denotes the language of all $k$-clauses. 
We find that $\sat{k}$ is preserved by every minimal operation on
level $k' > k$ (in particular, by $\near_{k+1}$); not preserved by any
operation on a level $k' \leq k$; and that any sign-symmetric language
$\Gamma$ which is not preserved by the $k$-universal partial operation
$\universal_k$ can qfpp-define $\sat{k}$. Assuming SETH, 
the minimal non-trivial pSDI-operations that preserve
$\Gamma$ therefore appear to be reasonable proxies for the complexity of
$\SAT(\Gamma)$. 

Finally, for each level $k$, there is a language -- namely the
language of roots of polynomials of degree less than $k$ -- which is
preserved by $\universal_k$ but not by any other operation at level
$k' \leq k$, and which does admit an improved
algorithm~\cite{LokshtanovPTWY17SODA}.
This shows that any ``dichotomy'' characterizing sign-symmetric
languages $\Gamma$ for which $\SAT(\Gamma)$ admits an improved
algorithm under SETH, cannot require a minimal non-trivial
pSDI-operation other than $\universal_k$ for some $k$. 

It remains to show that these very mild restrictions,
of requiring only the presence of a single non-trivial pSDI-operation
$f$ preserving $\Gamma$, can be powerful enough to ensure
that $\SAT(\Gamma)$ admits an improved algorithm. 
This is our topic of study for the next section. 

\section{Upper bounds for sign-symmetric satisfiability problems}
\label{section:upper}


In this section, we consider the feasibility of designing an improved
algorithm directly for \textsc{$\inv(f)$-SAT} and \textsc{$\inv(f)$-CSP}
for a minimal non-trivial pSDI-operation $f$, i.e., an improved
algorithm that only uses the abstract properties guaranteed by such an
operation $f$. 

We show this unconditionally for $f=\edge_2$ and for $f=\near_3$,
over arbitrary finite domains (where the latter result is only
interesting for the non-Boolean case, since the Boolean case is in
P).  The algorithms for these cases use, respectively, 
a \textsc{Subset Sum}-style meet-in-the-middle algorithm and
fast matrix multiplication over exponentially large matrices. 
These algorithms all work in the extension oracle model. 

We also show conditional or partial results.  We show two conditional
results for partial $k$-NU operations, showing that
\textsc{$k$-NU-CSP} admits an improved
algorithm in the oracle model if the \textsc{$(k, k-1)$-hyperclique}
problem admits an improved algorithm, 
and that \textsc{$k$-NU-SAT} admits an improved algorithm
in the explicit representation model if 
the Erd\H{o}s-Rado \emph{sunflower conjecture}~\cite{JLMS:JLMS0085}
holds for sunflowers with $k$ sets. The first of these results is a
direct generalisation of the matrix multiplication strategy; the
second uses fast local search in the style of
Sch\"oning~\cite{schoning1999}. 
Finally, we also consider the symmetric special case of
\textsc{3-edge-SAT}, and show that this problem reduces to a problem
of finding a unit-coloured triangle in an edge-coloured graph. 
This, in turn, follows from fast algorithms for sparse triangle
detection.  Several of the algorithms we reduce to have a running time
that depends on the matrix multiplication exponent $\omega$; the best
currently known value is $\omega < 2.373$~\cite{Gall14ammult,VWilliams12mmult}.

Before we begin, we need the following lemma, which shows that if a
relation is preserved by a pSDI-operation, then it is possible to view
the relation as a relation of smaller arity over a larger domain,
which is preserved by the corresponding partial operation over the
larger domain.

\begin{lemma} \label{lemma:larger}
  Let $R$ be an $n$-ary relation over a set of values $D$, $P$ a
  polymorphism pattern, and $f$ a partial operation preserving $R$ and
  satisfying $P$. Let $I_1
  \ldots, I_m$ be a partition of $[n]$, and $R_{I_1, \ldots, I_m}$ the
  $m$-ary relation \[R_{I_1, \ldots, I_m} = \{(\pro_{I_1}(t), \ldots,
  \pro_{I_m}(t)) \mid t \in R\}\] 
  over the set of values $\{\pro_{I_1}(R) \cup
  \ldots \cup \pro(I_m)(R)\}$. Then every partial operation $f'$
  satisfying $P$ over $\{\pro_{I_1}(R) \cup
  \ldots \cup \pro(I_m)(R)\}$ preserves $R_{I_1, \ldots, I_m}$
\end{lemma}

\begin{proof}
  Let $k = \ar(f') = \ar(f)$. Let $t_1, \ldots, t_k \in R$ and let
  $t'_1, \ldots, t'_k \in R_{I_1, \ldots, I_m}$ be the corresponding
  tuples of $R_{I_1, \ldots, I_m}$.
  Assume that $f'(t_1, \ldots, t_k)$ is defined, i.e., $(t_1[j],
  \ldots, t_k[j]) \in \domain(f')$ for each $j \in [k]$. Let $i \in
  [n]$ and let $I_j$ be the index set such that $i \in I_j$. Since
  $f'(t_1[j], \ldots, t_k[j])$ is defined it must be an instantiation
  of a tuple $p \in
  P$. It follows that $(t'_1[i], \ldots, t'_k[i])$ must be an instantiation
  of $p$ as well, implying that $f(t'_1[i], \ldots, t'_k[i])$ is
  defined. Hence, $f'$ preserves $R_{I_1, \ldots, I_m}$.
\end{proof}

\subsection{An $O^*(|D|^{\frac{n}{2}})$ algorithm for 2-edge-CSP}
\label{section:malt}
Given a binary relation $R$ one can construct a bipartite graph where
two vertices $x$ and $y$ have an edge between them if and only if
$(x,y) \in R$. Formally, the vertices $V_1 \cup V_2$ of this graph
will consist of the disjoint union of $\pro_1(R)$ and $\pro_2(R)$,
i.e., $V_1 = \{(1, x) \mid x \in \pro_1(R)\}$ and $V_2 = \{(2, x) \mid
x \in \pro_2(R)\}$. However, whenever convenient, we will not make
this distinction and instead assume that $V_1 = \pro_1(R)$ and $V_2 =
\pro_2(R)$.  We say that a binary relation $R$ is {\em rectangular} if
its bipartite graph representation is a disjoint union of bicliques.

\begin{lemma} \label{lemma:rect}
  Let $\malt_D$ be the partial Maltsev operation over a domain $D$. Then every
  binary relation preserved by $\malt_D$ is rectangular.
\end{lemma}

\begin{proof}
  The proof is very similar to the total case, which is essentially
  folklore in universal algebra. First note that $R$ is rectangular if
  and only if a path of length 4 between nodes $x, x', y, y'$ implies
  that there is an edge between $x$ and $y'$.  Therefore, let $(x,y),
  (x',y), (x',y') \in R$. But then $\malt_D((x,y), (x',y), (x',y')) =
  (\malt_D(x,x',x'), \malt_D(y,y,y')) = (x,y')$, implying that $(x,y') \in
  R$ since $R$ is preserved by $\malt_D$. Hence, $R$ is rectangular.
\end{proof}

If $R$ is an $n$-ary relation, $I_1 \cup I_2$ a partition of $[n]$,
and $s \in \pro_{I_1}(R)$, $t \in \pro_{I_2}(R)$, we write
$\concat{s}{t}{I_1}{I_2}$ to denote the $n$-ary tuple in $R$ satisfying
$\pro_{I_1}(\concat{s}{t}{I_1}{I_2}) = s$ and
$\pro_{I_2}(\concat{s}{t}{I_1}{I_2}) = t$. Let $D = \{d_0, d_1, \ldots,
d_{k-1}\}$ be a finite set of values. We can then order $D$ according
to a total order $<$, by letting $d_0 < d_1 < \ldots < d_{k-1}$. 
This order easily extends to $n$-ary tuples $s$ and $t$ over $D$ by
letting $s < t$ if and only if there exists an $i \in [n]$ such that
$\pro_{1, \ldots, i}(s) = \pro_{1, \ldots, i}(t)$ and $s[i+1] <
t[i+1]$. Given a relation $R$ we say that the tuple $t$ is {\em
  lex-min} if $t \in R$ and there does not exist any $t' \in R$ such
that $t' \neq t$ and $t' < t$.

\begin{lemma} \label{lemma:graph}
  Let $R$ be an $n$-ary relation preserved by $\malt_D$ and let $I_1
  \cup I_2$ be a partition of $[n]$. Then  there exists a bipartite
  graph $(V,E)$ where $V$ is the disjoint union of $\pro_{I_1}(R)$ and $\pro_{I_2}(R)$
  such that
  \begin{enumerate}
  \item
    $(V,E)$ is a disjoint union of bicliques,
  \item
    $\{s,t\} \in E$ if and only if
    $\concat{s}{t}{I_1}{I_2} \in R$,
  \item
    for every $s \in V$ occurring in a biclique $C_1 \cup C_2$ a pair
    $s_0 \in C_1, t_0 \in C_2$
    such that $s_0$ is lex-min in $C_1$ and $t_0$
    lex-min in $C_2$ can be computed in $O(\mathrm{poly}(n,|D|))$ time in
    the extension oracle model. 
\end{enumerate}
\end{lemma}

\begin{proof}
  Consider the binary relation $R_{I_1, I_2} = \{(\pro_{I_1}(t),
  \pro_{I_2}(t)) \mid t \in R\}$ over the set of values $\pro_{I_1}(R)
  \cup \pro_{I_2}(R)$. By Lemma~\ref{lemma:larger} this relation is
  preserved by $\malt$ over the larger domain, and Lemma~\ref{lemma:rect} then
  implies that $R_{I_1, I_2}$ is rectangular. Take the bipartite graph
  representation $(V_1 \cup V_2, E)$ of $R_{I_1, I_2}$ (which by the
  rectangularity property is a disjoint union of bicliques), and thus
  satisfies property (1). Property number (2) then follows easily from the
  construction of the bipartite graph $(V_1 \cup V_2, E)$ since two
  vertices $s$ and $t$ are connected with an edge if and only if
  $(s,t) \in R_{I_1, I_2}$, which holds if and only if $\concat{s}{t}{I_1}{I_2}
  \in R$. 

  For property (3) we need to show that we, given $s \in V$, can compute 
  lex-min representatives of the biclique $C_1 \cup C_2$
  containing $s$, in
  polynomial time with respect to $n$ and $|D|$. Assume without loss
  of generality that $s \in V_1$, and
  order $I_2$ in ascending order as $i_1, \ldots, i_{|I_2|}$. Then
  determine the smallest value $d_1 \in D$ such that
  $\concat{s}{(d_1)}{I_1}{\{i_1\}}$ is included in the projection
  $\pro_{I_1 \cup \{i_1\}}(R)$. This can be computed in polynomial time
  using the extension oracle. Then continue, by for each $i_2,
  \ldots, i_j$ determine the smallest $d_j \in D$ such that
  $\concat{s}{(d_1)}{I_1}{\{i_1\}} \in \pro_{I_1 \cup \{i_1, \ldots,
    i_j\}}(R)$. Let $t_0$ denote the resulting tuple, and observe that
  $t_0 \in C_2$ and that $\{s, t_0\} \in E$. We then repeat this using
  the index set $I_1$ in order to obtain a lex-min tuple $s_0$
  such that $\{s_0, t_0\} \in E$, which again can be done in
  polynomial time in the extension oracle model.
\end{proof}

\begin{theorem} \label{thm:2edgecsp}
  2-edge-CSP is solvable in
  $O^*(|D|^{\frac{n}{2}})$ time in both the extension oracle
  model and the explicit representation.
\end{theorem}

\begin{proof}
Let $(V,C)$ be an instance of 2-edge-CSP, where $V = \{x_1,
\ldots, x_n\}$ and $C = \{C_1, \ldots, C_m\}$. Assume without loss of
generality that $n$ is even, and let $I = [\frac{n}{2}]$ and $J = [n]
\setminus I$. Consider two sets $P$ and $Q$ constructed as
follows. Initially we let $P$ and $Q$ consist of all $\frac{n}{2}$-ary
tuples over $D$. Then, for each $p \in P$, $q \in Q$ we enumerate each
constraint in the instance containing only variables indexed by $I$ or
$J$ and check whether $p$ or $q$ is contradicted by the constraint. If
this is the case we remove $p$ from $P$ or $q$ from $Q$.  More
formally, if $p \in P$ and $R_i(x_{i_1}, \ldots, x_{i_k}) \in C$, $k =
\ar(R_i)$, such that $\{i_1, \ldots, i_k\} \subseteq I$, we check
whether $\pro_{i_1, \ldots, i_k}(p) \in \pro_{i_1, \ldots, i_k}(R_i)$, and similarly for $q \in
Q$. Each such step can be done in $O(\mathrm{poly}(k))$ time in the
extension oracle model and in $O(k + |R_i|)$ time if constraints
are explicitly represented. By repeating this for all elements in $P$ and
$Q$ we will therefore obtain two sets of partial assignments that do
not directly contradict individual constraints in the input instance.

Next, for each $p \in P$ and $q \in Q$ create two $m$-ary tuples $p'$
and $q'$. By using Lemma~\ref{lemma:graph} we for each constraint $C_i
\in C$ will associate the $i$th element of $p'$ and $q'$ with a
representative of the biclique corresponding to $C_i$, $p$, and $q$.
Hence, let $C_i = R_i(x_{i_1}, \ldots, x_{i_k}) \in C$, $k =
\ar(R_i)$, be a constraint. We distinguish between two cases. First,
assume that $\{i_1, \ldots, i_k\} \subseteq I$ or that $\{i_1, \ldots,
i_k\} \subseteq J$. In this case we for every $t \in P \cup Q$ let
$t'[i] = 1$. Second, assume that $i_1, \ldots, i_k \in I \cup J$ but
that $\{i_1, \ldots, i_k\} \not \subseteq I$ and $\{i_1, \ldots,
i_k\} \not \subseteq J$. In other words the constraint contains
variables indexed by members of both $I$ and $J$. For every $p \in P$
compute the lex-min representatives $p_0$ and $q_0$ of the biclique
containing $p$, with respect to the two index sets $P_i = \{j \mid i_j
\in I\}$ and $Q_i = \{j \mid i_j \in J\}$. This can be done in
polynomial time via Lemma~\ref{lemma:graph}. Assign the $i$th value to
the tuple $p'$ the value $(p_0, q_0)$, and then repeat this for every
$q \in Q$.

Let $P' = \{p' \mid p \in P\}$ and $Q' = \{q' \mid q \in Q\}$ be the
sets resulting from repeating this for every constraint in the
instance. We observe that the combination of $p \in P$ and $q \in Q$ satisfies a
constraint $R_i(x_{i_1}, \ldots, x_{i_k}) \in C$ if and only if $p'[i]
= q'[i]$, due to property (2) in Lemma~\ref{lemma:graph}. Hence, the
instance is satisfiable if and only if the two sets $P'$ and $Q'$
intersect. Since $P'$  and $Q'$ contain at most $|D|^{\frac{n}{2}}$
tuples, each of length $m$, this test can easily be accomplished in
$O^{*}(|D|^{\frac{n}{2}})$ time using standard algorithms.
\end{proof}

\subsection{An $O^{*}(|D|^{\frac{\omega n}{3}})$ algorithm for \textsc{3-NU-CSP}}

\label{section:maj}

The algorithm in Section~\ref{section:malt} used the rectangularity
property of binary relations in order to obtain an improved algorithm
for 2-edge-CSP. In this section we will devise an
$O^{*}(|D|^{\frac{\omega n}{3}})$ time algorithm for
3-NU-CSP 
by exploiting a structural property that is valid for all
ternary relations preserved by $\near_3$. 
Here, $\omega < 2.373$ is the matrix multiplication exponent.
We will need the following definition.

\begin{definition}
  An $n$-ary relation $R$ over $D$ is {\em $k$-decomposable} if there
  for every $t \notin R$ exists an index set $I \subseteq [n]$, $|I|
  \leq k$, such that $\pro_{I}(t) \notin \pro_{I}(R)$.
\end{definition}

In the total case it is known that $R$ is $k$-decomposable if $R$ is
preserved by a total $k$-ary NU-operation~\cite{jeavons1997}. In
general, this is not true for partial NU-operations, but we still
obtain the following result.

\begin{lemma} \label{lemma:decompose}
  Let $R$ be a $k$-ary relation preserved by $\near_k$. Then $R$ is $(k-1)$-decomposable.
\end{lemma}

\begin{proof}
  Let $t$ be a $k$-ary tuple not included in $R$. Assume that
  $\pro_{I}(t) \in \pro_I(R)$ for every index set 
  $I \subseteq [k]$, $|I| < k$. But then there must exist $t_1,
  \ldots, t_k \in R$ such that each $t_i$ differ from $t$ in at most
  one position. This furthermore implies that $\near_k(t_1, \ldots,
  t_k)$ is defined, and therefore also that $\near_k(t_1, \ldots, t_k)
  = t \notin R$. This contradictions the assumption that $\near_k$
  preserves $R$, and we therefore conclude that there must exist an
  index set $I \subseteq [k]$ of size at most $k-1$, such that
  $\pro_I(t) \notin \pro_I(R)$. 
\end{proof}

\begin{theorem} \label{thm:3-nu-csp}
  $3$-NU-CSP is solvable in $O^{*}(|D|^{\frac{\omega n}{3}})$
  time in both the extension oracle model and the explicit representation,
  where $\omega < 2.373$ is the matrix multiplication exponent.
\end{theorem}

\begin{proof}
  Let $(V, C)$ be an instance of
  3-NU-CSP where $V = \{x_1, \ldots, x_n\}$ and $C = \{C_1, \ldots, C_m\}$. Partition $[n]$
  into three sets $I_1, I_2, I_3$ such that $|I_i| = \frac{n}{3}$ (or, if
  this is not possible, as close as possible). Let $F_1, F_2, F_3$
  denote the set of all partial truth assignments corresponding to
  $I_1, I_2, I_3$, and observe that $|F_i| \leq |D|^{\frac{n}{3}}$. First,
  for each partial truth assignment $f \in F_i$, remove it from the set $F_i$ if there exists a
  constraint in the instance which is not satisfied by $f$. This can
  be done in polynomial time with respect to the number of constraints
  in the instance,  using a extension oracle query for each
  constraint. Second, construct a 3-partite graph where the node set is
  the disjoint union of $F_1$, $F_2$ and $F_3$, and add an edge between
  two nodes in this graph if and only if the combination of this partial truth
  assignment is not contradicted by any constraint in the
  instance. Last, answer yes if and only if the resulting graph contains
  a triangle.

  We begin by proving correctness of this algorithm and then analyse
  its complexity. We first claim that if the combination of $f_1 \in
  F_1, f_2 \in F_2, f_3 \in F_3$ does not satisfy a constraint in the
  instance, then there exists $g_1, g_2 \in F_1 \cup F_2 \cup F_3$ which
  do not satisfy the instance either. Hence, take a constraint
  $R(x_{i_1}, \ldots, x_{i_k}) \in C$, $k = \ar(R)$, which is not
  satisfied by the combination of $f_1, f_2, f_3$. Let $I'_1 = \{j
  \mid i_j \in I_1\}$, $I'_2 = \{j \mid i_j \in I_2\}$, and $I'_3 = \{j
  \mid i_j \in I_3\}$ and consider the relation
  $R_{I'_1, I'_2, I'_3} = \{(\pro_{I_1}(t), \pro_{I_2}(t),
  \pro_{I_3}(t)) \mid t \in R\}$ over the set of values $\pro_{I_1}(R)
  \cup \pro_{I_2}(R) \cup \pro_{I_3}(R)$. By Lemma~\ref{lemma:larger}
  this relation is preserved by the $\near_k$ operation over the larger
  domain, and it then follows from Lemma~\ref{lemma:decompose} that this
  relation is $2$-decomposable.
  But
  then it is easy to see that there must exist partial truth
  assignments $y_1, y_2 \in F_1 \cup F_2
  \cup F_3$ such that $y_1$ and $y_2$ do not satisfy
  $R(x_{i_1}, \ldots, x_{i_k})$. Hence, if $(V,C)$ is satisfiable, then there
  clearly exists a triangle in the 3-partite graph, and if there
  exists a triangle, then by following the reasoning above, the
  instance must be satisfiable.

  For the complexity, we begin by enumerating the three sets of
  partial truth assignments, which takes $O(|D|^{\frac{n}{3}})$
  time. We then remove any partial truth assignment which is not
  consistent with the instance, which increases this by a polynomial
  factor, depending only on the number of constraints and the
  extension queries for each constraint. Similarly, when
  constructing the 3-partite graph we enumerate all binary combinations of
  partial truth assignments from the three sets and check whether they
  are consistent. After this we check for the
  existence of a triangle in the resulting graph with $O(|D|^{\frac{n}{3}})$
  nodes, which can be solved in $O(|D|^{{\frac{n}{3}}^{\omega}})$ time
  for $\omega < 2.373$, using fast matrix
  multiplication.
\end{proof}

\subsection{Strategies for $k$-NU-SAT}

It is easy to see that the strategy used in Theorem~\ref{thm:3-nu-csp} extends to reducing 
\textsc{$k$-NU-CSP} problems to \textsc{$(k, k-1)$-hyperclique}, 
i.e., the problem of finding a $k$-vertex hyperclique in a
$(k-1)$-regular hypergraph.  Thus we get the following. 

\begin{lemma}
  Assume that \textsc{$(k, k-1)$-hyperclique} on $n$ vertices can be
  solved in time $O^*(n^{k-\varepsilon})$ for some $\varepsilon > 0$. 
  Then \textsc{$k$-NU-CSP} admits an improved algorithm in the
  extension oracle model, i.e., an algorithm running in time
  $O^*(|D|^{(1-\varepsilon')n})$ on domain size $D$ and on $n$
  variables, for some $\varepsilon' > 0$. 
\end{lemma}

However, it should be noted that this is a notoriously difficult
problem, and there is some evidence against such results~\cite{williams2018}.
Thus, we also investigate a less general algorithm that rests on a
milder assumption. 

\subsubsection{$k$-NU-SAT via local search}

We show that subject to a popular conjecture, \textsc{$k$-NU-SAT}
admits an improved algorithm in the explicit representation model via
a local search strategy.
To state this we need a few basic definitions. A \emph{sunflower (with $k$ sets)} 
is a collection of $k$ sets $S_1$, \ldots, $S_k$
with common intersection $S=S_1 \cap \ldots \cap S_k$,
called the \emph{core}, 
such that for every pair $i, j \in [k]$, $i \neq j$,
we have $S_i \cap S_j = S$. Note that we may have $S=\emptyset$.
The \emph{sunflower conjecture}~\cite{JLMS:JLMS0085}, in the form we will need,
states that for every $k$ there is a constant $C_k$
such that for every $n$, every collection of at least
$C_k^n$ sets of cardinality $n$ contains a sunflower with $k$ petals. 
This conjecture was the subject of the Polymath 10 collaborative
mathematics project, but remains a notorious open problem.
See Alon, Shpilka and Umans~\cite{AlonSU13sunflower}
for variations of the conjecture and connections to other problems.

We first show a simple connection between the sunflower conjecture for
sunflowers with $k$ sets and relations $R \in \inv(\near_k)$. 
For convenience, for a set $S \subseteq [n]$ we 
denote by $\chi_S^n$ the tuple $t \in \{0,1\}^n$
such that for each $i \in [n]$, 
$t[i]=1$ is $i \in S$ and $t[i]=0$ otherwise. 

\begin{lemma} \label{lemma:minl-sunflowerfree}
  Let $R \subset \{0,1\}^n$ be a relation with $0^n \notin R$.
  Say that a tuple $t=\chi_S^n$ is \emph{minimal in $R$} 
  if $t \in R$ but for every  $S' \subset S$ we have 
  $\chi_{S'}^n \notin R$. For $i \in [n]$, let $\mathcal{F}_i$
  be the set of minimal tuples in $R$ of Hamming weight $i$.
  If $R$ is preserved by $\near_k$, then $\mathcal{F}_i$ 
  does not contain a sunflower of $k$ sets. 
\end{lemma}
\begin{proof}
  Let $\mathcal{F}_i$ be as in the statement, and assume that
  $R$ is preserved by $\near_k$. Assume that 
  there are distinct sets $S_1, \ldots, S_k$ forming
  a sunflower with some core $S$, such that
  $\chi_{S_j} \in \mathcal{F}_i$ for every $j \in [k]$. 
  But then the operation $\near_k(\chi_{S_1}, \ldots \chi_{S_k})$
  is defined, and produces the tuple $\chi_S$. 
  This contradicts that the tuples are minimal in $R$. 
\end{proof}

We show that the sunflower conjecture is sufficient to allow an
improved algorithm. 

\begin{lemma} \label{lemma:sunflower-knualg}
  Assume that the sunflower conjecture holds for sunflowers with $k$
  sets, with some constant $C_k$. 
  Let $\Gamma$ be a sign-symmetric language preserved by $\near_k$. 
  Assume that for every $n$-ary relation $R \in \Gamma$
  and every $p \in [n]$, the minimal tuples in $R$
  of Hamming weight at most $p$ can be enumerated in time
  $O^*(2^{O(p)})$. 
  Then $\SAT(\Gamma)$ admits an improved algorithm. 
\end{lemma}
\begin{proof}
  We first show that the assumptions are sufficient to allow 
  a solution for the \emph{local search} problem for $\SAT(\Gamma)$,
  in the following form. Let an instance $(V,C)$ of $\SAT(\Gamma)$
  with $|V|=n$, a tuple $t \in \{0,1\}^n$, and an integer $p \in [n]$ 
  be provided. We can in $O^*(2^{O(p)})$ time decide whether there is
  a tuple $t' \in \{0,1\}^n$ 
  with Hamming distance at most $p$ from $t$
  that satisfies $(V,C)$.

  For this, we repeatedly perform the following procedure. 
  Verify whether the present tuple $t$ satisfies $(V,C)$,
  and if not, let $R(X)$ be a constraint in $C$ falsified by $t$, and
  let $I \subseteq [n]$ be the set of indices corresponding to the set
  of variables $X$.
  Let $s$ be the sign pattern such that $(\pro_I(t))^s = 0^{|X|}$. 
  Note that $R^s \in \Gamma$ by assumption. We then enumerate
  the minimal tuples in $R^s$ of Hamming weight at most $p$,
  and for every such tuple $t'$, of weight $i$, let $t''$ be
  the tuple $t$ with bits flipped according to $t'$, and 
  recursively solve the local search problem from tuple $t''$
  with new parameter $p-i$. Correctness is clear, since the search
  is exhaustive (because we loop through all minimal tuples).  
  We argue that this solves the local search problem itself in
  $O^*(2^{O(p)})$ time. For the running time, assume for simplicity
  that producing the tuples takes $O^*(c^p)$ time and, for the same
  constant $c$, there are at most $c^i$ minimal tuples of weight $i$
  (by Lemma~\ref{lemma:minl-sunflowerfree}). Up to polynomial factors, 
  the running time is then bounded by a recurrence
  \[
  T(p) = c^p + \sum_{i=1}^p c^iT(p-i),
  \]
  which is bounded as $T(p) \leq (2c)^p$. 
  
  From here on, well-known methods can be used to complete 
  the above into an improved algorithm; cf. Sch\"oning's algorithm for
  $k$-SAT~\cite{schoning1999} and its derandomization~\cite{DantsinGHKKPRS02}, 
  or even restrict the above to \emph{monotone} local search instead of
  arbitrary local search and apply the method of
  Fomin et al.~\cite{FominGLS16localsearch}.  
\end{proof}

In particular, this is allows for an algorithm in the explicit
representation model.  

\begin{theorem} \label{theorem:sunflower-algorithm}
  Assume that the sunflower conjecture holds for sunflowers with $k$
  sets. Then \textsc{$k$-NU-SAT} admits an improved algorithm in 
  the explicit representation model. 
\end{theorem}

We leave it as an open question whether access to an extension oracle
(also known as an interval oracle)
suffices to solve the local search problem in single-exponential
time. The problem, of course, is that the bounds above only apply to
the \emph{minimal} tuples, and while it is easy to find a single
minimal tuple using an extension oracle, it is less obvious how to 
test for the existence of a minimal tuple within a given interval.
Meeks~\cite{Meeks16IPEC} showed how a similar result is possible,
but her method would require an oracle for finding minimal satisfying
tuples of weight \emph{exactly} $i$, which is also not clear how to do.

\subsubsection{$k$-NU-SAT and bounded block sensitivity}

Finally, we briefly investigate connections between the $\near_k$
partial operation and a notion from Boolean function analysis known as
\emph{block sensitivity}, introduced by Nisan~\cite{Nisan91SICOMP}.
See also the book by O'Donnell~\cite{ODonnellBook2014}.

We first introduce some temporary notation.
For any relation $R \subseteq \{0,1\}^n$,
let $f_R: \{0,1\}^n \to \{0, 1\}$ be a function
defined as $f_R(t)=[t \in R]$, i.e., $f_R(t)=1$
if $t \in R$ and $f_R(t)=0$ otherwise. 
For a tuple $t \in \{0,1\}^n$ and 
a set $S \subseteq [n]$, let $t^S$ 
denote the tuple $t$ with the bits of $S$ flipped.
A function $f:\{0,1\}^n \to \{0,1\}$
has \emph{block sensitivity} bounded by $b$
if for every $t \in \{0,1\}^n$ there are at 
most $b$ disjoint sets $S_1, \ldots, S_b \subseteq [n]$
such that $f(t^{S_i})\neq f(t)$ for every $i \in [b]$.
We show that $\near_k$ can be seen as a one-sided version of block
sensitivity. 

\begin{lemma}
  Let $R \subseteq \{0,1\}^n$ be a relation.
  Then $f_R$ has block sensitivity less than $k$
  if and only if both $R$ and its 
  complement $\overbar{R} := \{0,1\}^n \setminus R$
  are preserved by $\near_k$. 
\end{lemma}
\begin{proof}
  In the first direction, assume that $f$ has block sensitivity 
  at least $k$. Let $t \in \{0,1\}^n$ be a tuple and let $[n]=X_0 \cup
  \ldots \cup X_k$ be a partition of $[n]$ into blocks such that for
  each $1 \leq i \leq k$,  we have $f(t^{X_i}) \neq f(t)$. 
  Then if $f(t)=1$, then the tuples $t^{X_i}$ form a witness
  that $R$ is not preserved by $\near_k$, and if $f(t)=0$ they
  form a witness against $\overbar{R}$ being preserved by $\near_k$. 
  In the other direction, let $t_1, \ldots, t_k \in R$
  be such that $\near_k(t_1, \ldots, t_k) = t$ is defined
  and $t \notin R$. For $i \in [k]$, let $X_i$ be the positions $j$
  where $t[j] \neq t_i[j]$. Then $X_1 \cup \ldots \cup X_k$ forms
  a subpartition of $[n]$, showing that $f$ has block sensitivity
  at least $k$. The case that $\overbar{R}$ is not preserved
  by $\near_k$, instead of $R$, is completely dual.
\end{proof}

It is known that a block sensitivity of at most $b$ implies a
\emph{certificate complexity} of at most $b^2$, i.e., for any relation
$R \in \inv(\near_k)$ and any tuple $t \in \overbar{R}$, there are at
most $b^2$ bits in $t$ that certify that $t \notin R$~\cite{Nisan91SICOMP}. 
This suggests a branching or local search algorithm for $\SAT(\Gamma)$
where $\Gamma$ contains such relations. However, more strongly, it
implies that $R$ has a decision tree of bounded
depth~\cite{Nisan91SICOMP}, and thus, since $k$ is a 
constant, that $R$ only depends on constantly many arguments. 
Thus, block sensitivity is a significantly stronger
restriction than what $\near_k$ imposes.

However, one related question remains. Assume that $R$ is an $n$-ary
relation preserved by $\near_k$, and which does depend on all its
arguments. Is there a non-trivial upper bound on $|R|$, e.g.,
does it hold that $|R| \leq (2-\varepsilon_k)^n$ for some
$\varepsilon_k$ depending on $k$? A positive answer to this question
would imply a trivial improved algorithm for \textsc{$k$-NU-SAT} via
enumeration of satisfying assignments, constraint by constraint.


\subsection{Symmetric 3-edge-SAT}

We finish this section with a result showing that a number of special
cases of \textsc{3-edge-CSP} admits an improved algorithm via sparse
triangle finding. The class in particular contains 
\textsc{3-edge-SAT} for symmetric relations $R \in \inv(\edge_3)$. 
We begin by characterising the symmetric relations in
$\inv(\edge_3)$. 

\begin{lemma} \label{lemma:cases-symmetric-3edge}
  Let $R \subseteq \{0,1\}^n$ be a symmetric relation preserved by $\edge_3$,
  Let $S \subseteq \{0,\ldots,n\}$ be the weights accepted by $R$. 
  Then either $S$ is a complete arithmetic progression (possibly a
  trivial one, of length 1), or $S=\{a,a+b\}$ or $S=\{n-a, n-a-b\}$
  for some $a<b$. 
\end{lemma}
\begin{proof}
  Let us first make a simpler claim: 
  If $a, a+b \in S$ is a pair that does not extend to a complete
  progression in $S$, then either $a-b<0$ or $a+2b > n$. 

  To see this, let $a, a+b \in S$, and assume
  $a+2b \notin S$, $a+2b \leq n$. 
  First assume $a \geq b$. We subpartition $[n]$ into one set 
  $T_0$ of size $a-b \geq 0$  and three sets $T_i$ of size $b$, $i=1, 2, 3$.
  This is possible since $a-b+3b =a+2b\leq n$.
  Let $t=\chi_{T_0 \cup \ldots \cup T_3}$ and for $i=1, 2, 3$ let $t_i=t^{T_i}$.
  Finally, let $t_4=t^{T_1 \cup T_2}$.
  Then $\edge_3(t_1,\ldots,t_4)$ is defined and produces $t$.
  Thus we conclude $a<b$, i.e., $a-b < 0$. 
  By the symmetric argument, if $a, a+b \in S$
  with $a-b \geq 0$ and $a-b \notin S$, then $a+2b > n$. 
  This finishes the claim.

  Next, assume that $|S|>2$ and that $S$ contains some pair $a, a+b$
  such that the progression does not continue. Let $b>0$ be the
  smallest value such that such a pair exists, and again by symmetry 
  assume that $a+2b \leq n$; thus $a-b<0$. Let $c \in S \setminus \{a,a+b\}$. 
  First assume $c > a+2b$. Then we may, similarly to above,
  pack sets with $|T_0|=a$, $|T_1|=|T_2|=b$, and $|T_3|=c-a-2b$,
  and we have a witness showing $a+2b \in S$.  But in the remaining
  cases, $c$ must be involved in a complete progression with either
  $a$ or $a+b$, by the choice of $a$ and $b$. It is easy to check that
  this implies  the existence of a value $c' \in S$ with $a < c' < a+b$, 
  and that iterating the claim eventually produces an arithmetic
  progression of step size dividing $b$, covering $a$ and $a+b$,
  contradicting the assumption that $a+2b \notin S$. 
  Thus $|S|=2$, i.e., $S=\{a,a+b\}$. 
\end{proof}

In particular, this lemma shows that every symmetric relation in
$\inv(\edge_2)$ is a simple arithmetic progression. It also shows that
$R$ has a simple-to-compute \emph{2-edge embedding}, i.e., $\hat R
\supseteq R$, $\hat R \cap \{0,1\}^{\ar(R)} = R$, and $\hat R$ is
preserved by a total 2-edge operation~\cite{LagerkvistW17CP}, produced
by extending $S$ into a complete progression.

We now describe the algorithm. Let $R$ be a relation with arguments $X$.
For a partition $X=X_1 \cup X_2$ and an assignment $f$ to $X_1$, 
we refer to the \emph{2-edge label} of $f$ as the pair
$(f_0, g_0)$ produced by first extending $f$ to a lex-min assignment
$g_0$ such that $(f,g_0) \in R$, then extending $g_0$ to a lex-min
assignment $f_0$ such that $(f_0, g_0) \in R$. Note that this is the
same procedure used in the algorithm for \textsc{2-edge-CSP}.

We extend this to 3-partite graphs as follows. Let the variable
set be partitioned as $[n] = X \cup Y \cup Z$, and define a graph
$G=(V,E)$ with partition $V=V_X \cup V_Y \cup V_Z$, where the nodes of
each part represent partial assignments as in Section~\ref{section:maj}.
For each edge, verify that the corresponding partial assignment is 
consistent with each relation in the input instance.  
We proceed to give labels to edges of $G$ for each relation $R$ as
follows. We assume that for each relation, the ``type'' of $R$ is
known to us (2-edge, 3-NU, or symmetric 3-edge). 
If $R \in \inv(\near_3)$, all edges get the same label. 
Otherwise, let $\hat R \supseteq R$ be the 2-edge-embedding of $R$
(with $\hat R=R$ if $R$ is already 2-edge). 
Let $pq$ be an edge in $G$, corresponding to partial assignments $p, q$. 
If one of these assignments, say $p$, is an assignment to $X$, then we 
set the label of $pq$ to the 2-edge label of $p$
in the partition $X \cup (Y \cup Z)$. Otherwise, $p \cup q$ is an 
assignment to $Y \cup Z$, and we set the label of $pq$
to the 2-edge label of this assignment in $X \cup (Y \cup Z)$.
We show that this label scheme captures our language. 

\begin{lemma} \label{lemma:label-triangles-work}
  Let $R$ be a relation with arguments $U$, for some $U \subseteq [n]$, 
  and let $G=(V,E)$ and $X \cup Y \cup Z$ be as above. 
  If either $R \in \inv(\edge_2)$, or $R \in \inv(\near_3)$, 
  or $R$ is Boolean, symmetric and $R \in \inv(\edge_3)$,
  then a triple $(f, g, h)$ with $f \in V_X$, $g \in V_y$, $h \in V_z$
  satisfies $R$ if and only if $fgh$ is a triangle in $G$ where the 
  edges $fg$, $fh$, $gh$ all have the same label. 
\end{lemma}
\begin{proof}
  Refer to a triangle $fgh$ with all edge labels identical as a
  \emph{single-label triangle}. We will also slightly abuse notation
  by treating $R$ as a 3-ary relation taking values from $V_X \times
  V_Y \times V_Z$.  First assume that $R \in \inv(\edge_2)$, and recall that
  $R$ is rectangular. 
  Let $fgh$ be a single-label triangle with shared label $L=(f_0,g_0h_0)$; 
  we show that $(f, g, h) \in R$. Since $L$ is the label of the edge $gh$, 
  it must be that $(f_0, g, h), (f_0, g_0, h_0) \in R$, and by the edges $fg$ and
  $fh$ it must be that $(f, g_0, h_0) \in R$ as well. By the partial
  2-edge operation, this implies $(f, g, h) \in R$.
  Thus every single-label triangle corresponds to a satisfying assignment.
  
  In the other direction, let $(f, g, h) \in R$. Since $R$ is rectangular, 
  there is a unique lex-min pair $(f_0,  g_0h_0)$ in the biclique
  containing $(f, gh)$, and both extensions $(f, g_0h_0)$ and $(f_0, gh)$ 
  are compatible with $R$. Thus all three edges get the same label and
  the algorithm works for $R \in \inv(\edge_2)$. 

  The case $R \in \inv(\near_2)$ is trivial. Since such a relation is
  2-decomposable, the entire verification of $R$ happens in the stage
  where edges are filtered, and in the remaining graph, every triangle
  represents a satisfying assignment and every triangle is single-label.

  Finally, assume $R \in \inv(\edge_3)$ and is symmetric. 
  If $R \in \inv(\edge_2)$, then we argue as above. Otherwise, 
  by Lemma~\ref{lemma:cases-symmetric-3edge}, 
  either $S=\{a,a+b\}$ or $S=\{n-a,n-a-b\}$ for $a<b$,
  and $\hat R$ verifies that each assignment $(f,g,h)$
  has the correct weight when computed $\mathrm{mod}\, b$. 
  First assume that $fgh$ is a single-label triangle in $G$. 
  First assume $S=\{a,a+b\}$. By the edge-filtering step, we know
  that for each of the edges $fg$, $gh$, $fh$
  the corresponding partial assignment has weight at most $a+b$. 
  Thus the total weight of $(f,g,h)$ is at most $(a+b)(3/2) \leq 
  a+b+(a+b)/2 < a+2b$. Dually, assume $S=\{n-a-b, n-a\}$. No edge in
  $fgh$ has more than $a+b$ zeroes, thus the total assignment has
  weight greater than $n-a-2b$. In both case, since the edge-labels
  work to verify the value $\mathrm{mod}\, b$, we conclude $(f,g,h) \in R$. 

  On the other hand, assume $(f,g,h) \in R$. Since the edge labels
  verify the more permissive relation $\hat R$, the triangle $fgh$ is
  a single-label triangle. 
\end{proof}

The remaining problem can now be solved via algorithms for
triangle-finding in sparse graphs. 

\begin{theorem} \label{thm:3-edge--almost}
  Assume a CSP or SAT problem with the following characteristic: 
  for every relation $R$, either $R \in \inv(\edge_2)$ and $R$
  is labelled with type $\edge_2$, or $R \in \inv(\near_3)$ and 
  $R$ is labelled with type $\near_3$, or the language is Boolean,
  $R$ is a symmetric relation in $\inv(\edge_3)$ and
  $R$ is labelled with type $\edge_3$. 
  This problem can be solved in time
  $O^*(|D|^{\frac{\omega + 3}{6}})$
  in the extension oracle model,
  where $\omega < 2.373$ is the matrix multiplication exponent. 
\end{theorem}
\begin{proof}
  By the description above, we create a 3-partite graph $G$
  on $3|D|^{n/3}$ vertices (where $|D|=2$ in the Boolean case), 
  and for every edge in $G$ we give it a vector of labels, 
  one label per relation in the input instance. We refer to this
  vector as the \emph{colour} of the edge. 
  Note that a symmetric relation $R$ can be ``inspected''
  using its extension oracle to find out the set $S$ of accepted
  weights. By Lemma~\ref{lemma:label-triangles-work}, the instance
  has a satisfying assignment if and only if $G$ has a triangle where
  all edges have the same colour. 
  
  This we solve as follows. For every colour $c$ used by an edge in
  $G$, we generate the graph $G_c$ consisting of all edges of colour
  $c$. Let $m_c$ be the number of edges of $G_c$, and let
  $N \leq 3|D|^{n/3}$ be the number of vertices in $G$. 
  We check if $G_c$ contains a triangle. If $G_c$ is dense enough,
  then we use the  usual triangle-finding algorithm for this, with running time
  $O^*(N^{\omega})$, otherwise we use an algorithm for triangle
  finding in sparse graphs. 
  Alon, Yuster and Zwick~\cite{AlonYZ97cycles}  show such an algorithm 
  with running time $O(m_c^{2\omega/(\omega+1)})$, 
  where $\omega < 2.373$ is the matrix multiplication exponent. 
  Hence, the crossover point at which we use the dense algorithm
  is $m_c \geq N^{(\omega + 1)/2}=:N^\alpha$. 
  Summing over all colours, we have $\sum_c m_c \leq N^2$.
  Since the algorithm for sparse graphs has a super-linear running 
  time, the worst case is when we are at the crossover density
  and use the sparse algorithm $N^{2-\alpha}$ times for a cost of
  $O(N^{\omega})$ each time. This works out to a total running time
  $O(N^{(\omega+3)/2})$ for triangle-finding, i.e., the CSP is solved
  in time $O^*(|D|^{(\omega+3)n/6}) = O^*(|D|^{0.896n})$ using
  $\omega=2.373$.
\end{proof}

We do not know whether this strategy can be extended to arbitrary
relations $R \in \inv(\edge_3)$, even for a non-uniform algorithm. 

\paragraph{Section summary.}
We have proven that it is indeed feasible to construct improved
algorithms for $\inv(p)$-SAT and $\inv(p)$-CSP for individual
pSDI-operations $p$. A crucial step for constructing algorithms of
this form is first to identify non-trivial properties of relations
invariant under $p$, which for the partial 2-edge operation turned out
be rectangularity, and for the partial 3-NU operation
$2$-decomposability. However, it might not always be the case that
every invariant relation satisfies such a clear-cut property, and for
3-edge-SAT we had to settle for an improved algorithm for symmetric
relations.

For $k$-NU-CSP and $k$-NU-SAT we also gave conditional improvements in
terms of \textsc{$(k, k-1)$-hyperclique} and the sunflower
conjecture. At the present, it is too early to say whether these
algorithms constitute the only source of improvement or if more direct
arguments are applicable.

\section{Lower Bounds}
\label{section:lower}
In this section we turn to the problem of proving lower bounds for
sign-symmetric $\SAT$ problems.

\subsection{Lower bounds based on $k$-SAT}

As an easy warm-up, we first consider languages $\Gamma$ such that
$\SAT(\Gamma)$ is at least as hard as $k$-SAT for some $k$. 
For each $k \geq 3$ let $c_k \geq 0$ denote the infimum of the set $\{c
\mid k$-SAT is solvable in $O(2^{cn})$ time$\}$. Under the ETH, $c_k >
0$ for each $k \geq 3$, and for each $k \geq 3$ there exists $k' > k$
such that $c_{k'} > c_k$~\cite{impagliazzo2001}. 
The best known upper bounds yield $c_k \leq 1-\Theta(1/k)$, but
no methods for lower-bounding the values $c_k$ are known.

Recall that Lemma~\ref{lemma:notkuniv-kcnf} gives a condition under
which a language $\Gamma$ can qfpp-define all $k$-clauses. We
observe the immediate consequence of this.

\begin{lemma} \label{lemma:reducefromksat}
  Let $\Gamma$ be a sign-symmetric constraint language not preserved
  by the $k$-universal partial operation. Then $\SAT(\Gamma)$ cannot
  be solved in time $O^*(2^{cn})$ for any $c<c_k$, even in the
  non-uniform model. 
\end{lemma}
\begin{proof}
  By Lemma~\ref{lemma:notkuniv-kcnf}, $\Gamma$ can qfpp-define all
  $k$-clauses. More concretely, there is a finite set $\Gamma'
  \subseteq \Gamma$ of relations such that every $k$-clause has a
  fixed, finite-sized gadget implementation over $\Gamma'$. Thus,
  given a $k$-SAT instance on $n$ variables, we can produce an
  equivalent instance of $\SAT(\Gamma')$ in linear time, with the same
  variable set. 
\end{proof}

As a consequence, $c_k$ is also a lower bound on the running time for
$\inv(f)$-SAT for every minimal pSDI-operation at level $k+1$ and
higher. However, this above lemma applies to any sign-symmetric
constraint language, and not just to the special case when $\Gamma=\inv(f)$.
We can also observe a similar consequence for SETH-hardness. 

\begin{corollary} \label{corollary:musthavekuni}
  Let $\Gamma$ be a sign-symmetric constraint language not preserved
  by the $k$-universal partial operation for any $k$. Then assuming
  SETH, $\SAT(\Gamma)$ does not admit an improved algorithm, even in
  the non-uniform model. 
\end{corollary}
\begin{proof}
  By SETH, there is for every $\varepsilon > 0$ a constant $k$
  such that $k$-SAT cannot be solved in $O^*((2-\varepsilon)^n)$
  time. By Lemma~\ref{lemma:reducefromksat}, there is a reduction
  from $k$-SAT to $\SAT(\Gamma)$ for this $k$. Thus, $\SAT(\Gamma)$ 
  does not admit an improved non-uniform algorithm. 
\end{proof}

\subsection{2-edge-SAT and Subset Sum}
\label{section:subsetsumhard}

Next, we sharpen the connection between \textsc{Subset Sum} and
\textsc{2-edge-SAT}.  Recall that an instance of \textsc{Subset Sum}
consists of a set $S=\{x_1, \ldots, x_n\}$ of $n$ numbers and a target
integer $t$, with the question of whether there is a set $X' \subseteq
S$ such that $\sum X' = t$. This can also be phrased as
asking for $z_1, \ldots, z_n \in \{0,1\}$ such that
\[
\sum_{i=1}^n z_i x_i = t.
\]
Also recall from Lemma~\ref{lemma:2edge:linear} that such a relation is contained in $\inv(\edge_2)$. 
However, this does not by itself imply a problem reduction, since 
an instance or \textsc{2-edge-SAT} assumes the existence of an
extension oracle for every constraint. We show that such a reduction
can be implemented by splitting the above equation apart into several
equations, based on the bit-expansion of $t$. 

\begin{theorem}
  If \textsc{2-edge-SAT} is solvable in $O(2^{cn})$ time for $c > 0$
  in the extension oracle model, 
  then \textsc{Subset Sum} is solvable in $O(2^{(c + \varepsilon)n})$ time for every $\varepsilon > 0$.
\end{theorem}

\begin{proof}
  Let $x_1, \ldots, x_n, t \in \N$ be the input to a \textsc{Subset Sum}
  instance. We will reduce this instance in subexponential time
  to a disjunction over \textsc{2-edge-SAT} instances on $n$ variables
  each.
  
  We proceed as follows. Harnik and Naor~\cite{doi:10.1137/060668092} give a randomized procedure
  for this that reduces a \textsc{Subset Sum} instance to bit length
  at most $2n + \log \ell$, where $\ell$ is the bit length of the
  input. If $\ell \geq 2^n$, then we solve the instance by brute force
  in time polynomial in the input length, otherwise we are left with
  an instance of bit length $\ell' \leq 3n$. 
  
  Next, set a parameter $k=\sqrt{n}$, and split the binary expansion
  of the input integers into $k$ blocks of equal length, giving
  $\sqrt{n}$ blocks of length  $O(\sqrt{n})$. For each block
  guess the contribution of the solution to the target value. Note
  that the maximum overflow that can carry over to the next block is
  $n$, which means that for a single block there are $O(n^2)$ options
  for the contribution within the block. 
  We get at most $O(n^{2k})=2^{o(n)}$ guesses in total, after which we
  have replaced the original equation $\sum_i z_ix_i=t$ by the
  conjunction of $\sqrt{n}$ linear equations, each with a target
  integer of $O(\sqrt{n})$ bits. 
  This allows us to implement an extension oracle for every
  such constraint with a running time of $2^{O(\sqrt{n})}$,
  using the well-known tabulation approach.
  
  This encodes an instance of \textsc{2-edge-SAT} in the extension
  oracle model with $n$ variables. Using an algorithm for this
  problem, and multiplying its running time by the time required for
  answering an oracle query, yields the claimed running time for
  \textsc{Subset Sum}.  
\end{proof}

Given that the running time for \textsc{2-edge-SAT} in the extension
oracle model given in this paper matches the best known running time
for \textsc{Subset Sum}, and given that improving the latter is a
long-open problem,  it seems at the very least that an improvment to
\textsc{2-edge-SAT} would require significant new ideas.

\subsection{Padding formulas}
\label{section:padding_constants}

We now give a combinatorial interlude, showing how relations $R
\subseteq \{0,1\}^n$ can be padded with additional variables such that
the new relation lies in $\inv(f)$, for any non-total partial
operation~$f$. This will be leveraged in the next section to finally
provide concrete lower bounds on the running time of $\inv(f)$-SAT for
pSDI-operations $f$. 

For a partial operation $p$, say of arity $k$, and a sequence of tuples
$t_1, \ldots, t_k$, we say that $p(t_1, \ldots, t_k)$ is a
\emph{projective application} if $p(t_1, \ldots, t_k)$ is either
undefined or $p(t_1, \ldots, t_k) \in \{t_1, \ldots, t_k\}$. 
Similarly, if $p(t_1, \ldots, t_k)$ is defined and 
$p(t_1, \ldots, t_k) \notin \{t_1, \ldots, t_k\}$ 
we call $p(t_1,\ldots, t_k)$ a {\em non-projective application}.

\begin{definition}
  Let $R \subseteq \{0,1\}^n$ be a relation and $P$ a set of Boolean partial operations. 
  A \emph{padding} of $R$ with respect to $P$
  is an $(n + m)$-ary relation $\padd{R}$ such that (1) $\pro_{1, \ldots,
    n}(\padd{R}) = R$,
  (2) $|\padd{R}| = |R|$, and (3) $\padd{R} \in \inv(P)$. 
  A \emph{universal padding formula} for $n \geq 1$ with respect to $P$
  is an $(n+m)$-ary relation $\upadd{P}$ which (1) is a padding of the
  relation $\{0,1\}^n$ and (2) $p(t_1, \ldots, t_{\ar(p)})$ is a projective
  application for every partial operation $p \in P$ and every sequence of tuples $t_1, \ldots, t_{\ar(p)} \in \upadd{P}$.
\end{definition}

Note that if $R$ is a relation and $p$ a $k$-ary partial operation
such that $p(t_1, \ldots, t_k)$ is a projective application for every
sequence $t_1, \ldots, t_k \in R$, then $R \in \inv(P)$. In particular
this implies that $\upadd{P} \in \inv(P)$ for every universal padding
formula $\upadd{P}$ of $P$. Also, critically, if 
$\upadd{P}$ is an $(n+m)$-ary universal padding formula for a set of
partial operations $P$, and $R$ is an $n$-ary relation, then the
relation $R'(x_1, \ldots, x_n, y_1, \ldots, y_m) \equiv R(x_1, \ldots,
x_n) \land \upadd{P}(x_1, \ldots, x_n, y_1, \ldots, y_m)$ is a padding
formula for $R$. Hence, a universal padding formula can be viewed as a 
blueprint which can be applied to obtain a concrete padding formula
for any relation. 
It is known that if $P$ contains no total operation, then a universal
padding formula can be constructed using a universal hash
family~\cite{lagerkvist2017c}. 



\begin{lemma} \label{lemma:padding_constant}
  Let $P$ be a finite set of partial operations such that the only
  total functions in $\strongof{P}$ are projections. For
  every $n \geq 1$ there exists an $(n+m)$-ary universal padding formula $\upadd{P}$
  such that $m \leq c \cdot n$, for a constant $c$ depending on $P$.
\end{lemma}

\begin{proof}
  See Lagerkvist \& Wahlstr\"om~\cite[Lemma~35]{lagerkvist2017c}.
\end{proof}

A quick note is in place on the role of universal padding formulas in 
obtaining lower bounds for $\inv(P)$-SAT, when $P$ is a finite set of
partial operations. Note that in a standard ``gadget'' reduction
from CNF-SAT to some problem $\SAT(\Gamma)$, one would introduce 
some number of local variables for every clause of the input, to
create an equivalent output formula that only uses constraints from
$\Gamma$. The existence of padding formulas does allow us to do this
for $\inv(P)$-SAT, but for lower bounds under SETH this is not useful
since we have no control over the number of additional variables
created this way. However, the universality property of universal
padding formulas allow us to \emph{reuse} the padding variables
between different constraints, to produce an output which only has
$n+m=O(n)$ variables in total. The details are given in the next section,
but first we investigate concrete values of the
constant $c$ for partial $k$-edge and $k$-NU operations.


\begin{lemma}
  Let $X=\{x_1, \ldots, x_n\}$ be a set of variables,
  and let $y=\bigoplus_{i \in S} x_i$ be the parity sum 
  for a set $S \subseteq [n]$ chosen uniformly at random. 
  For any tuple $t \in \{0,1\}^n$, let $t'$ be $t$ padded by $y$.
  Let $p$ be a partial operation as specified below, let $r=\ar(p)$,
  and let $(t_1, \ldots, t_r)$ be a sequence of tuples in $\{0,1\}^n$
  such that $p(t_1,\ldots,t_r)$ is a non-projective application.
  Then the following hold.
  \begin{enumerate}
  \item If $p$ is the partial 2-edge operation, with $r=3$, 
    then the probability that $p(t_1', t_2', t_3')$ is defined
    is $3/4$.
  \item If $p$ is the partial 3-edge operation, with $r=4$,
    then the probability that $p(t_1', \ldots, t_4')$ is defined
    is $1/2$.
  \item If $p$ is the partial $k$-NU operation, $k \geq 4$, then the
    probability that $p(t_1', \ldots, t_r')$ is defined is $(2k+2)/2^k$.
    For every weaker operation, e.g., for the partial $k$-edge or
    $k$-universal operations, the probability is at most this high. 

  \item If $p$ is the partial $k$-universal operation, $k \geq 3$, 
    then the probability that $p(t_1', \ldots, t_r')$ is defined
    is $(k+1)/2^k$.
  \end{enumerate}
\end{lemma}
\begin{proof}
  Throughout the proof, we write $y(t)=\bigoplus_{i \in S} t[i]$. 
  Let us consider each case in turn.

  \emph{1.} We have $\ar(p)=3$. Let $I$ respectively $J$ be the set of
  indices $i \in [n]$ such that $t_1[i]=t_2[i] \neq t_3[i]$, respectively,
  $t_1[i] \neq t_2[i] = t_3[i]$. Note that both $I$ and $J$ are non-empty
  since $p(t_1,t_2,t_3)$ is a non-projective application. 
  Then $p(y(t_1), y(t_2), y(t_3))$ is undefined if and only if
  the parity of $S \cap I$ and $S \cap J$ are both odd. Since $I$ and $J$
  are disjoint, the probability of this is exactly $1/4$. 

  \emph{2.} For the partial 3-edge operation, recall from Theorem~\ref{lemma:weakeststrongestminimal} that $p$ can be constructed
  by adding a fictitious argument to the partial 3-NU operation.
  Hence, the arguments $i \in [n]$ such that $(t_1[i], \ldots, t_4[i])$
  is non-constant partition into three sets $I_1, I_2, I_3 \subseteq [n]$,
  and since $p(t_1,\ldots, t_4)$ is a non-projective application, all three sets
  must be nonempty. It can be verified that $p(y(t_1), \ldots, y(t_4))$
  is defined if and only if $S \cap I_i$ is odd for at most one $i \in [3]$.
  This happens with exactly $1/2$ probability. 
%

  \emph{3.} For the partial $k$-NU operation, we have $\ar(p)=k$; 
  let $p(t_1,\ldots,t_k)=t$. There are $k$ non-empty pairwise disjoint 
  sets $I_1$, \ldots, $I_k$ such that $t_i[j] \neq t$ if and only if 
  $j \in I_i$, for each $i \in [k]$, $j \in [n]$.
  The tuple $(y(t_1), \ldots, y(t_k))$ has one value, say $b$,
  in every row $i \in [k]$ where $S \cap I_i$ is odd,
  and another value, $1-b$, in every row $i$ where $S \cap I_i$ is odd.
  Thus $p(y(t_1), \ldots, y(t_k))$ is defined if either
  $S \cap I_i$ is odd for at most one index or
  $S \cap I_i$ is even for at most one index;
  these are $2k+2$ possibilities. For all other
  $2^k-(2k+2)$ possibilities, the operation is undefined.
  Note that all these possibilities happen with 
  equal probability, since the sets $I_i$ are non-empty 
  and pairwise disjoint. 

  \emph{4.} We have $\ar(p)=2^k-1 = r$, with the non-constant
  parts of $\domain(p)$ partitioned into $k$ pairs. 
  Let $I_i$, $i \in [k]$ be the sets of indices $j \in [n]$ 
  such that $(t_1[j], \ldots, t_r[j])$ belongs to the $i$th
  of these pairs, in some enumeration. 
  We claim that $p(y(t_1),\ldots, y(t_r))$ is defined
  if and only if $S \cap I_i$ is odd for at most one $i \in [k]$.
  On the one hand, if this holds, then 
  $(y(t_1), \ldots, y(t_r))$ is contained in pair number $i$ 
  or is constant, and it is clear that the operation is defined. 
  Otherwise, let $S \cap I_i$ and $S \cap I_j$ both be odd, $i \neq j$. 
  Let $t=p(t_1,\ldots,t_r)$; let $a \in [r]$ be the argument 
  such that $t_a[i] \neq t[i]$ if and only if $i \in I_i$;
  let $b \in [r]$ be the argument such that $t_b[i] \neq t[i]$
  if and only if $i \in I_j$; and let $c \in [r]$ be the
  argument such that $t_c[i] \neq t[i]$ if and only if 
  $i \in I_i \cup I_j$. Then the three positions $y(t_a), y(t_b), y(t_c)$
  have a pattern that is not compatible with any domain
  element of $p$. 
  It follows that the probability that $p(t_1',\ldots,t_r')$
  is defined is exactly $(k+1)/2^k$. 
\end{proof}


\begin{lemma} \label{lemma:countapplications}
  Let $p$ be a partial operation. There are $(|\domain(p)|)^n$
  sequences $(t_1, \ldots, t_{\ar(p)})$ of tuples in $\{0,1\}^n$
  such that $p(t_1,\ldots, t_{\ar(p)})$ is defined.
\end{lemma}
\begin{proof}
  For every argument $i \in [n]$, we choose which element
  from $\domain(p)$ the tuple $(t_1[i], \ldots, t_{\ar(p)}[i])$
  will correspond to. Every such choice results in a distinct
  sequence of tuples. 
\end{proof}

\begin{lemma} \label{lemma:paddingbounds}
  Let $R(x_1, \ldots, x_n, y_1, \ldots, y_m)$ be a padding formula for
  $\{0,1\}^n$, where each $y_i$ is a
  a parity bit over $\{x_1, \ldots, x_n\}$ chosen uniformly at random.
  Then the following hold. 
  \begin{enumerate}
  \item For the partial 2-edge operation, $R(x_1, \ldots, x_n, y_1, \ldots, y_m)$ is a universal padding formula
    with probability at least $1-\varepsilon$ if $m \geq 6.23n+\log (1/\varepsilon)$. 
  \item For the partial 3-edge operation, $R(x_1, \ldots, x_n, y_1, \ldots, y_m)$ is a universal padding formula
    with probability at least $1-\varepsilon$ if $m \geq 3n+\log (1/\varepsilon)$.
  \item For the partial $k$-NU operation, $k \geq 4$, and for any operation weaker than it,
    $R(x_1, \ldots, x_n, y_1, \ldots, y_m)$ is a universal padding formula with exponentially small failure probability
    if $m = \Omega(\frac{\log k}{k} n)$. 
  \end{enumerate}
\end{lemma}
\begin{proof}
  \emph{1.} 
  By Lemma~\ref{lemma:countapplications}, there are $6^n$ triples
  such that $p$ is defined. For each such triple such that the
  application of $p$ is non-projective, the probability 
  that it remains defined after the addition of a single random parity 
  bit is $3/4$. Thus after adding $t$ parity bits, the expected number 
  of  non-projective triples is at most
  \[
  6^n (3/4)^t=2^{n \log 6 - t \log (4/3)}.
  \]
  With $t = (n \log 6)/(\log 4/3) + d$, this number equals $1/2^d$,
  which means that with probability at least $1-1/2^d$, no defined triples remain.  
  The constant factor works out to
  $(\log 6)/(\log (4/3)) = (1+\log 3)/(2 - \log 3) < 6.23$. 
  
  \emph{2.} There are $8^n$ tuples $(t_1,\ldots,t_4)$ such that
  $p$ is defined, and for each of them which is non-projective the probability 
  of remaining defined after the addition of a single parity
  bit is $1/2$. Thus adding $3n+d$ parity bits leaves in expectation
  at most
  \[
  8^n (1/2)^{3n+d} = 2^{-d}
  \]
  non-projective tuples, and the probability that no non-projective tuples
  remain is at least $1-1/2^d$. 

  \emph{3.} In the general case, there are $(2k+2)^n=2^{(1+\log (k+1))n}$ 
  defined tuples, and the probability of a non-projective tuple remaining defined after the 
  addition of a random parity bit is $O(k/2^k)$.
  Note that $(ck/2^k)^t = 2^{(\log c + \log k - k)t}$. 
  Thus the expected number of non-projective tuples after $t$ parity
  bits is at most
  \[
  2^{(1+\log(k+1))n-(k - \log k - c')t},
  \]
  and it suffices to let $t=\Omega(\frac {\log k}{k} n)$.
\end{proof}

We remark that with a padding strategy other than simple parity bits,
a significantly lower scaling ratio may be possible for the
partial $k$-universal operation.  However, the advantage of paddding
with parity bits is that the padding can be efficiently inverted,
allowing for efficient extension oracles for the padded relation.

\subsection{Lower bounds in the extension oracle model}
\label{section:oracle_bounds}


In this section we use the bounds obtained in
Section~\ref{section:padding_constants} to obtain lower bounds for
$\inv(P)$-SAT in the extension oracle model.

\begin{lemma} \label{lemma:parityextensionoracle}
  Let $\upadd{P}$ be an $(n+m)$-ary universal padding formula via the construction
  in Lemma~\ref{lemma:paddingbounds}. Let $R = \{0,1\}^{k} \setminus \{t\}$
  for a $k$-ary tuple $t \in \{0,1\}^{k}$. Then there is a polynomial-time extension
  oracle for $R(x_1, \ldots, x_k) \land \upadd{P}(x_1,
  \ldots, x_n, y_1, \ldots, y_m)$. 
\end{lemma}
\begin{proof}
  Let $\alpha : X \rightarrow \{0,1\}$, $X \subseteq \{x_1, \ldots, x_k,
  y_1, \ldots, y_m\}$, be a partial truth assignment. We need to show
  that we can decide if $\alpha$ is consistent with $R(x_1, \ldots, x_k)
  \land \upadd{P}(x_1, \ldots, x_n, y_1, \ldots, y_m)$ in polynomial
  time. First, we check whether $\alpha$ is consistent with the
  constraint $R(x_1, \ldots, x_k)$, which is easy to do due to the
  representation of $R$. Second, recall that there for each $y_i$ exists an
  index set $S_i$ such that $y_i = \bigoplus_{s \in S_i} x_s$. Hence,
  the partial assignment $\alpha$ together with $R(x_1, \ldots, x_k)
  \land \upadd{P}(x_1, \ldots, x_n, y_1, \ldots, y_m)$ induces a system
  of  linear equations over GF(2) where the unknown variables are those
  unassigned by $\alpha$.  We may thus solve this system and check
  whether it has any solution $f$ where $f[i] \neq t[i]$ for some $i \in [k]$.
\end{proof}




\begin{theorem} \label{thm:concrete-lower-bounds}
  Let $P$ be a set of partial operations, and set $m \geq cn + \log n$ 
  such that a random parity-padded formula $\upadd{P}(x_1, \ldots,
  x_n, y_1, \ldots, y_m)$ is a universal padding formula with high
  probability. 
  Then $\inv(P)$-SAT cannot be solved in time $O^*(2^{(1/(c+1)-\varepsilon)n})$
  for any $\varepsilon>0$, assuming the randomized version of the 
  SETH is true. 
  In particular, we have the following lower bounds for specific problems:
  \begin{enumerate}
  \item 2-edge-SAT cannot be solved in $O(2^{(c-\varepsilon)n}$ time
    for any $\varepsilon > 0$, where $c \approx 1/7.28$.
  \item 3-edge-SAT cannot be solved in $O(2^{(c-\varepsilon)n}$ time
    for any $\varepsilon > 0$, where $c = 1/3$.
  \item For $k \geq 4$, $k$-NU-SAT cannot be solved in $O(2^{(c-\varepsilon)n})$ time
    for any $\varepsilon > 0$, where $c=1-\Theta(\frac{\log k}{k})$,
    and the same bound holds for the harder problems
    $k$-edge-SAT and $k$-universal SAT.    
  \end{enumerate}
\end{theorem}
\begin{proof}
  Let $\mathcal{F}$ be a CNF-SAT instance on variable set $X$, $|X|=n$,
  and compute a random padding formula $\upadd{P}(x_1, \ldots, x_n,
  y_1, \ldots, y_m)$, with $m$ as stated.
  We assume that the construction 
  is successful,
  i.e., that the resulting relation is a universal padding formula
  with respect to $P$. For every clause in the input,
  defined on a tuple of variables $(x_{i_1}, \ldots, x_{i_r})$, let
  $R(x_{i_1}, \ldots, x_{i_r})$ be the 
  corresponding relation, and let $R'(x_{i_1}, \ldots, x_{i_r}) \land \upadd{P}(x_1,
  \ldots, x_n, y_1, \ldots, y_m)$ be the relation as in
  Lemma~\ref{lemma:parityextensionoracle} (up to the ordering of
  variables). Note that we do not need
  to explicitly enumerate the tuples in this relation, since we may simply
  provide the extension oracle proven to exist
  in Lemma~\ref{lemma:parityextensionoracle}. 
  Then the output is a conjunction of $\inv(P)$-SAT relations,
  with a polynomial-time extension oracle for each one,
  and the resulting instance is equivalent to $\mathcal{F}$. 
  Since the output instance has $n + m = (c+1) \cdot n$ variables,
  an algorithm solving $\inv(P)$-SAT faster than the time stated
  would imply an improved algorithm for CNF-SAT.
  The bounds for specific problems follow from the bounds for universal
  padding formulas computed in Lemma~\ref{lemma:paddingbounds}.
\end{proof}


Finally, we note that the convergence of the lower bounds for
$k$-NU-SAT towards $2^n$, assuming SETH, is at a slower rate 
than the upper bounds for the best known algorithms for $k$-SAT,
which scale as $c_k \leq 1-\Theta(1/k)$~\cite{impagliazzo2001}.
There are also significant differences in problem model (finite
language versus infinite language, and concrete constraints versus
extension oracles). It would be interesting to improve these results,
to either improve the convergence rate or provide bounds in some
explicit representation model, assuming SETH. 

\paragraph{Section summary.}
We have proven lower bounds under SETH. The bounds obtained in
Theorem~\ref{thm:concrete-lower-bounds} are only valid in the extension oracle
model, and it does not appear entirely straightforward to extend them
to the explicit representation. However, for 2-edge-SAT we also gave a lower bound subject to
the \textsc{Subset Sum} problem, which as remarked is strong evidence
that the $O^{*}(2^{\frac{n}{2}})$ algorithm from Theorem~\ref{thm:2edgecsp} is the
best we could reasonably hope for.

\section{Discussions and Conclusions}

We have investigated the structure of constraint languages under
fine-grained reductions, with a focus on sign-symmetric Boolean
languages, and applied the results to an analysis of the time
complexity of NP-hard SAT problems, in a general setting.

The structural analysis uses an algebraic connection to analyse
constraint languages via their partial polymorphisms.  Thereby the
structural conclusions are relevant for any problem that takes as
input a constraint formula over some fixed constraint language, under
just a few assumptions: (1) that the constraints in the formula are
``crisp'' rather than soft, and are required to all be satisfied
(as opposed to problems such as MAX-SAT, where a feasible solution may
falsify some constraints); (2) that there are no structural
restrictions of the formula itself (e.g., no bounds on the number of
occurrences per variable); and (3) that the constraint language is
sign-symmetric, i.e., allows the free application of negated variables
and the use of constants in constraints.  Thus it naturally applies to
SAT$(\Gamma)$ problems, but would also be relevant for the analysis of
problems such as \#SAT and optimisation problems, or even 
parameterized problems such as \textsc{Local Search SAT$(\Gamma)$} --
is there a solution within distance $k$ of a given non-satisfying
assignment $t$?

\textbf{Structural results.}
The expressive power of sign-symmetric languages is
characterised by the restricted partial polymorphisms in this paper
referred to as pSDI-operations. We characterise the structure of all minimal
non-trivial pSDI-operations, and find that they are organised into a
hierarchy, whose levels correspond to the problem complexity, with
close connections to being able to express the $k$-SAT languages. 
Moreover, we described the weakest and strongest operations on each level. 
We find that particular families of pSDI-operations correspond to
partially defined versions of well-known algebraic conditions from the
study of CSPs; in particular, the strongest operation at each level
$k$ corresponds to the $k$-NU condition. Finally, we also give a
result in the ``vertical'' direction of the hierarchy, 
giving a simple characterisation of languages not preserved by the
partial $k$-NU operation for any $k$. By the above discussion, this
result should be of interest also for other inquiries. 

\textbf{Complexity of SAT$(\Gamma)$ problems.}
We apply our results to an analysis of the fine-grained time
complexity of $\SAT(\Gamma)$ for sign-symmetric languages, under SETH.
We consider previously studied languages with improved algorithms --
i.e., such that $\SAT(\Gamma)$ can be solved in time $O^*(c^n)$ for
some $c<2$ -- and find that they correspond well to particular classes
of the hierarchy. Conversely, every known language $\Gamma$
such that $\SAT(\Gamma)$ is SETH-hard -- i.e., admits no improved
algorithm assuming SETH -- lives entirely outside of the hierarchy.
We also show the feasibility of giving improved algorithms whose
correctness relies only and directly on the above-mentioned
pSDI-operations, by showing that known algorithmic strategies such as
fast matrix multiplication and (conjecturally) fast local search can
be extended to work for such classes.  

Finally, we give complementary lower bounds -- for every invariant $f$
as above, there is a constant $c_f$ such that \textsc{$\inv(f)$-SAT}
cannot be solved in $O^*(c^n)$ time for any $c<c_f$, assuming SETH.
These results are arguably the first of their kind; every previously
known concrete lower bound under SETH has either been for showing that
a problem admits no non-trivial algorithm, or has been applied to
problems analysed under more permissive parameters such as treewidth. 
In particular, \textsc{2-edge-SAT} is the first SAT problem which
simultaneously has non-trivial upper and lower bounds on the running
time under SETH. 

\subsection{The abstract problem and polynomial-time connections}

Finally, let us make a short detour to consider what we may call the
abstract problem. We have noted that for every Boolean pSDI-operation $f$, 
there is a set of equational conditions that characterise $f$,
similarly to definitions of varieties in universal algebra, and for every larger
domain $D$, these conditions will uniquely determine a partial
operation over the domain $D$. Furthermore, these conditions are
preserved under taking powers of the domain, which we have exploited
for particular cases of \textsc{$\inv(f)$-SAT} and \textsc{$\inv(f)$-CSP}
to reduce input instances to instances of polynomial-time solvable
problems on exponentially many variables. 

These polynomial-time problem will in general be search problems, like
CSPs, and will be preserved by the same type of operation $f$, 
but have a fixed number of variables $d$ and with an unbounded
domain size $n$. Let us refer to this as the \emph{abstract
  $\inv(f)$-problem}. 
The question can be raised, for which pSDI-operations $f$ 
does such a problem allow improved polynomial-time algorithms? 

We refrain from phrasing the question formally, because the
polynomial-time complexity may be strongly affected by details such as
constraint representation, but we note that the class of problems
defined this way, unlike the original problems $\SAT(\Gamma)$, contain
several problems conjectured \emph{not} to have such an improvement. 

First, we note that every constraint of arity less than $d$ is
preserved by the $k$-NU-type partial operation with $k \geq d$. 
This in particular includes the $k$-hyperclique problem for
$(k-1)$-uniform hypergraphs, which has been conjectured not to be
solvable in time~$O(n^{k-\varepsilon})$ for any $\varepsilon>0$ and $k>3$~\cite{williams2018}. 
Thus the abstract $d$-NU problem does not admit an improved algorithm
for $d>3$ under this conjecture. 

Second, it can be verified that the problem of finding a zero-weight
triangle, under arbitrary large edge weights, if viewed as a single
constraint of arity $d=3$, is preserved by the corresponding
3-universal partial operation. It is known that subject to the 3SUM
conjecture, this problem cannot be solved in $O(n^{3-\varepsilon})$
for any $\varepsilon>0$~\cite{VWilliamsW13SICOMP}. 

If we restrict ourselves to the minimal non-trivial pSDI-operations
$f$ defined for the Boolean domain in this paper, this
leaves only a small number of concrete problems open
under the above conjectures. By the inclusions we have established, 
any operation $f$ at a level $k>3$ yields an abstract problem as hard
as the $k$-NU operation. Furthermore, the abstract 3-NU problem does
admit an improved algorithm via fast matrix multiplication. 
It can be easily checked that up to argument permutation, there are
only eight distinct pSDI-operations $f$ at level 3 of the hierarchy; 
and by the above discussion, the easiest and the hardest are
(conjecturally) resolved.  We consider it an interesting question 
to investigate the complexity of the problem for these  remaining
cases.  

\subsection{Regarding a dichotomy for sign-symmetric SAT problems}


Ignoring for the moment the lower bounds discussed in the previous
section, the results throughout our paper suggest a simple potential
dichotomy between NP-complete SAT problems solvable in $O(2^{cn})$
time for $c < 1$ and SAT problems not solvable in $O(2^{cn})$ time for
any $c < 1$ unless SETH fails. We can formulate this conjecture
as follows. To simplify the conjecture we restrict ourselves to the non-uniform model.

\begin{conjecture} \label{conjecture:1}
  Let $\Gamma$ be a possibly infinite sign-symmetric Boolean
  constraint language such that $\SAT(\Gamma)$ is NP-complete. Then
  $\SAT(\Gamma)$ admits a non-uniform algorithm with running time
  in $O(2^{cn})$ time for $c < 1$ if and only if $\Gamma$ is preserved
  by a non-trivial pSDI-operation.
\end{conjecture}

Note that by Corollary~\ref{corollary:musthavekuni}, the negative
direction of this conjecture is already known, up to SETH. It thus
remains to consider whether \textsc{$k$-universal SAT} admits a
non-uniform improved algorithm for every $k$. Furthermore, as
discussed in the Introduction, the class of constraints definable as
the roots of bounded-degree multivariate polynomials represents an
example which by Lemma~\ref{lemma:dpoly-d+1univ} is directly
associated with \textsc{$k$-universal SAT}, and which has an improved
algorithm by Lokshtanov et al.~\cite{LokshtanovPTWY17SODA}. 
Thus, the above conjecture at least represent a kind of
Occam's razor-type extrapolation of least mathematical surprise. 

%
%

However, at the moment this conjecture seems difficult to settle. 
An extreme negative result, such as the conclusion that the full
problem \textsc{$\inv(f)$-SAT} admits an improved algorithm only when
the abstract $\inv(f)$-problem does, would by
Theorem~\ref{theorem:sunflower-algorithm} need to refute the sunflower 
conjecture. A full positive resolution would need to generalise the
result of Lokshtanov et al.~\cite{LokshtanovPTWY17SODA} to apply based
only on a weak abstract condition, whereas their present algorithm
strongly uses properties specific to polynomials.  
Intermediate outcomes are of course possible, but would raise further
questions of which pSDI-operations $f$ are powerful enough 
to guarantee the existence of an improved algorithm.

\subsection{Future work}

The investigations in this paper leave several concrete open
questions, and significant avenues for future work, regarding all
parts of the paper. Let us highlight a few.

\textbf{Structural aspects.} 
Assuming that the class of partial $k$-edge operations turn out to be
relevant for the analysis of future problems, it would be valuable to
have a set of canonical consequences to a language not being
preserved by any partial $k$-edge operation, similarly to
Theorem~\ref{thm:not-any-knu}. 
To this aim, it may also be enlightening to fully describe the symmetric relations
contained in various classes in the hierarchy.

Another concrete question is regarding the structure of
$\inv(\near_k)$ for $k>3$. Assume that $R \in \inv(\near_k)$ is an
$n$-ary Boolean relation, which depends on every argument. Is there a
non-trivial upper bound on $|R|$?

\textbf{Extension to CSPs.} 
Many questions remain regarding an extension of the project to CSPs on
non-Boolean domains. While the minimal non-trivial pSDI-operations 
defined in this paper do have higher-domain analogues, via
polymorphism patterns, and while these analogues do in some cases have
useful consequences for the complexity of the corresponding CSP, 
it is not clear that they are in general the only kind of condition
that is relevant for the fine-grained complexity of CSPs. 
In particular, in the Boolean domain there is a known correspondence
between pSDI-operations and sign-symmetric languages. 
No such correspondence has been shown for CSPs in general. 

In a different vein, for higher-domain CSPs there are also classes of
NP-hard problems whose time complexity is far better than
$O^*(|D|^n)$, e.g., \textsc{$k$-Colouring} corresponds to a CSP of
domain size $|D|=k$ and can be solved in $O^*(2^n)$ time for every
$k$~\cite{BjorklundHK09}. Arguably, we do not have a good
understanding of when this occurs in general, and we cannot claim that
an $O(c^n)$ time algorithm for $c < |D|$ is necessarily an
improvement. A reasonable starting point to mitigate some of these
technical difficulties is to initially only consider consider
constraint languages whose total polymorphisms are the projections.


\textbf{Problems.}
Let us mention a few concrete algorithmic questions. 
First of all, by Lemma~\ref{lemma:sidon-uni3}, symmetric relations
defined by Sidon sets are preserved by the $3$-universal operation,
but they do not seem to be captured by currently known algorithms for
problems in this class. Does the language consisting of all such
relations admin an improved algorithm?

Another problem is to find a generalisation of the algorithm for
constraints defined via bounded-degree
polynomials~\cite{LokshtanovPTWY17SODA},
without explicitly using properties specific to polynomials. 
A different generalisation of this class was considered
by the present authors (see the arXiv version of~\cite{LagerkvistW17CP}),
in the form of relations with bounded-degree Maltsev embeddings. 
Since this properly generalises bounded-degree polynomials,
it is natural to ask whether this class admits an improved
algorithm. 

More broadly, as remarked earlier, the classification of the
expressiveness of sign-symmetric constraint languages may be of
interest for questions other than just satisfiability.
The algorithm for \textsc{2-edge-SAT}, for instance, can be used
to solve the corresponding counting problem, showing that
pSDI-operations may be powerful enough also in other settings. 
Concrete questions to consider here include improved algorithms for
the counting problem $\#\SAT(\Gamma)$ and the parameterized problem
\textsc{Local search SAT$(\Gamma)$}. 

\textbf{Lower bounds.}
Can the padding scheme be improved to give better asymptotics with
respect to the level $k$?  Recall that the lower bound behaves
as a bound of $2-\Theta((\log k)/k)$, whereas all known algorithmic
strategies yield running times of the form $(2-\Theta(1/k))^n$. 

It would also be very interesting to have a SETH-based lower bound in
the explicit representation model. As discussed earlier the padding
construction is valid also in this representation, but is difficult to
implement in practice since the resulting relations may contain 
exponentially many tuples with respect to the number of variables.



\bibliography{references}
\bibliographystyle{abbrv}
\end{document}